\documentclass[final,3p,times,authoryear]{elsarticle}

\usepackage{amsmath}
\usepackage{amssymb}
\usepackage{listings}
\usepackage{amsthm}
\theoremstyle{plain}%
\newtheorem{theorem}{Theorem}[section]
\newtheorem{proposition}[theorem]{Proposition}%
\newtheorem{lemma}[theorem]{Lemma}
\newtheorem{corollary}[theorem]{Corollary}

\theoremstyle{definition}%
\newtheorem{definition}{Definition}%
\newtheorem{example}{Example}%

\theoremstyle{remark}%
\newtheorem{remark}{Remark}%

\theoremstyle{claim}%
\newtheorem{claim}[theorem]{Claim}%
\date{October 13, 2024}
\begin{document}

\begin{frontmatter}


\title{Graphon games and an idealized limit of large network games}


\author[inst1]{Motoki Otsuka}

\affiliation[inst1]{organization={University of Tsukuba},
            addressline={1-1-1}, 
            city={Tennodai},
            postcode={305-8577}, 
            state={Ibaraki},
            country={Japan}}



\begin{abstract}
\textit{Graphon games} are a class of games with a continuum of agents, introduced to approximate the strategic interactions in large network games. The first result of this study is an equilibrium existence theorem in graphon games, under the same conditions as those in network games. We prove the existence of an equilibrium in a graphon game with an infinite-dimensional strategy space, under the continuity and quasi-concavity of the utility functions. The second result characterizes Nash equilibria in graphon games as the limit points of asymptotic Nash equilibria in large network games. If a sequence of large network games converges to a graphon game, any convergent sequence of asymptotic Nash equilibria in these large network games also converges to a Nash equilibrium of the graphon game. In addition, for any graphon game and its equilibrium, there exists a sequence of large network games that converges to the graphon game and has asymptotic Nash equilibria converging to the equilibrium. These results suggest that the concept of a graphon game is an idealized limit of large network games as the number of players tends to infinity.
\end{abstract}



\begin{keyword}
Graphon game \sep Network game \sep Continuum of agents \sep Existence of equilibria \sep Idealized limit property
\JEL C02, C62, C72
\end{keyword}

\end{frontmatter}
\textbf{Data availability}: We do not analyse or generate any datasets, because our work proceeds within a theoretical and mathematical approach.\\
\textbf{Declarations of interest}: none.

\section{Introduction}
\label{sec:sample1}
Network games have emerged as an important area in economics and game theory. These games analyze the interaction of individuals who are connected through a network and whose behaviors are influenced by those around them. These games are beneficial for studying peer effects in various contexts, such as education choices, criminal activities, or the adoption of new technologies (for a survey, see \cite{RefWorks:RefID:159-jackson2015chapter} and \cite{10.1093/oxfordhb/9780199948277.013.8}). However, there are some practical challenges. One significant challenge is that real-world social networks are often extremely large, making it difficult for researchers to obtain accurate information about the network structure. Collecting detailed network data can be costly and is not always feasible due to privacy concerns. Another challenge arises from the size of the social network itself. When analyzing these networks based on the concept of Nash equilibrium, researchers are required to solve high-dimensional optimization problems, which becomes increasingly difficult as the network size grows.

To address these issues, \cite{parise2023} introduce the concept of a \textit{graphon game}. A graphon is a measurable function $W:[0,1]^2 \to [0,1]$ and is introduced as the limits of convergent sequences of networks as the number of nodes tends to infinity.\footnote{Intuitively, a sequence of networks $\{G_n\}$ is said to converge if, for any fixed network $F$, the proportion of $F$ in $G_n$ converges. The concept of graphons was introduced by \cite{RefWorks:RefID:409-lovász2006limits} and further developed by \cite{RefWorks:RefID:410-borgs2008convergent, RefWorks:RefID:411-borgs2012convergent}. The important results are well summarized in \cite{RefWorks:RefID:412-lovász2012large} and \cite{RefWorks:RefID:413-janson2013graphons}. The term "graphon" is short for "graph function."} A graphon can also be interpreted as a stochastic model for generating networks: independently and uniformly selecting a finite number of points $x_1, x_2, \cdots ,x_n$ from the interval $[0,1]$ and connecting $x_i$ and $x_j$ with probability $W(x_i,x_j)$.
In a graphon game, the space of agents is represented by the Lebesgue interval, and the graphon describes the connection strength between any pair of agents. Each agent’s payoff depends on their own action and a weighted average of other agents’ actions. These weights are heterogeneous and are specified by the graphon.

\cite{parise2023} make two main contributions to graphon games, among others. First, they prove the existence of a Nash equilibrium in graphon games under specific conditions on the utility functions, including continuous differentiability, strict concavity, and Lipschitz continuity. Second, they investigate large network games where agents interact according to a finite network sampled from a graphon and show that, under certain assumptions, the equilibria in such large sampled network games can be closely approximated by the unique equilibrium of the corresponding graphon game. This result is significant because graphon games generally involve lower-dimensional optimization problems compared to large network games. In summary, by interpreting graphons as stochastic models, \cite{parise2023} present a new framework for analyzing large network games.

In contrast, interpreting graphons as limiting objects allows us to view graphon games as the limits of network games as the number of players tends to infinity. This raises two natural questions:
\begin{enumerate}
    \item Can we establish the existence of an equilibrium in graphon games under the same conditions as in network games?
    In network games, it is well-known that a Nash equilibrium exists if (i) the strategy set is nonempty, convex, and compact, and (ii) the utility functions are continuous and quasi-concave. Moreover, the strategy space is allowed to be infinite-dimensional. In contrast, the conditions used by \cite{parise2023}, such as continuous differentiability and strict concavity, are stronger, and their analysis focuses on a finite-dimensional strategy space.
    \item Are equilibria of graphon games realizable as the “limits” of equilibria of large network games?
    That is, given a graphon game and an equilibrium of this game, does a large network game “similar” to the graphon game have an equilibrium “similar” to the equilibrium of the graphon game? 
\end{enumerate}
These two questions can be summarized into one central question: How ideal is the concept of a graphon game as a limiting object of network games as the number of players tends to infinity? We aim to address these questions.

The first result of this study is a unified equilibrium existence theorem for both network and graphon games. To address both types of games simultaneously, we consider graphon games with a measure space of agents and prove an equilibrium existence theorem under the assumptions of continuity, quasi-concavity of the utility functions, and an infinite-dimensional strategy space. This theorem includes, as corollaries, the equilibrium existence results for network games and graphon games with a continuum of agents, under the same assumptions on the utility function and strategy space. Thus, this result establishes the existence of an equilibrium in graphon games under the same conditions as in network games.

The second result characterizes Nash equilibria in graphon games as the limit points of asymptotic Nash equilibria in large network games. First, we prove that any graphon game can be viewed as the limit of a sequence of large network games. Next, we prove that if a sequence of large network games converges to a graphon game, any convergent sequence of asymptotic Nash equilibria in these large network games also converges to a Nash equilibrium of the limiting graphon game. Finally, for any graphon game and its equilibrium, there exists a sequence of large network games that converges to the graphon game and has asymptotic equilibria converging to the equilibrium. These results demonstrate that Nash equilibria in graphon games are the limit points of asymptotic Nash equilibria in large network games. Together with the first result, graphon games are, in this sense, ideal as limits of large network games.

This paper is organized as follows. In Section \ref{sec:literature}, we give a review of the related literature. Section \ref{sec:Notations and definitions} introduces the notations and definitions. We present our existence results in Section \ref{sec: Existence results} and the characterization result in Section \ref{sec:Characterization}. Section \ref{sec:concluding remarks} contains the concluding remarks. The omitted proofs are collected in the Appendix.

\section{Related literature}\label{sec:literature}
The study of games with a continuum of agents, often referred to as large games or non-atomic games, was initiated by \cite{RefWorks:RefID:419-schmeidler1973equilibrium} and \cite{RefWorks:RefID:402-mas-colell1984theorem}, the former in its individualistic form, and the latter in its distributional setting.
Traditional large games assumed that a player’s payoff depends on societal actions and their own actions, and thus did not account for players’ biological or social traits in understanding player interdependence. 
To address this issue in large games, \cite{RefWorks:RefID:273-khan2013large} introduce large games with traits, where a player's payoff is influenced by their own actions and social responses, which consider not only a summary of societal actions but also a summary of traits (see also \cite{RefWorks:RefID:424-khan2013large}, \cite{RefWorks:RefID:281-qiao2014space}, \cite{RefWorks:RefID:279-qiao2016closed-graph}, \cite{RefWorks:RefID:414-he2017modeling}, \cite{RefWorks:RefID:393-wu2022pure-strategy}, and \cite{RefWorks:RefID:70-he2022conditional}). Similarly, graphon games can model players' biological or social traits through the use of graphons. The main difference between these two concepts lies in the externalities experienced by players. In large games with traits, each player has a biological or social trait, and their payoff depends on the joint distribution of traits and actions. In contrast, graphon games use heterogeneous weights to calculate a weighted average of other agents' actions, so the externality each player faces varies from player to player. Thus, graphon games are better suited for analyzing the impact of social networks on game outcomes. In fact, we provide a method to construct Nash equilibria of graphon games with linear quadratic utilities from their graphons, as this is one of the most significant utility functions in the network game literature.

Since games with a continuum of agents are used as approximations of finite-player games, much research has focused on studying the relationship between large games and finite-player games. One way to relate large finite-player games and non-atomic games is to investigate whether the limit of a convergent sequence of equilibria in finite-player games is an equilibrium in the limiting non-atomic game. This issue is addressed by \cite{Green1984continuum}, \cite{RefWorks:RefID:28-keisler2009saturated}, \cite{RefWorks:RefID:281-qiao2014space}, \cite{RefWorks:RefID:279-qiao2016closed-graph}, \cite{RefWorks:RefID:414-he2017modeling}, and \cite{RefWorks:RefID:393-wu2022pure-strategy}. Another approach, known as \textit{asymptotic implementation}, starts with an equilibrium of the non-atomic game and asks whether sufficiently large finite-player games close to the non-atomic game have equilibria that are also close to that equilibrium. This approach is closely related to the characterization result in this study.

\cite{Hausman1988infinite} demonstrates the asymptotic implementation in convex games—defined by convex action sets and quasi-concave payoff functions—through approximate equilibria in large finite-player games.
Recently, \cite{RefWorks:RefID:122-carmona2020pure, RefWorks:RefID:394-carmona2021strict} have proved the asymptotic implementation in terms of exact Nash equilibria in large finite-player games without requiring convexity assumptions.
One limitation of \cite{RefWorks:RefID:122-carmona2020pure, RefWorks:RefID:394-carmona2021strict} is that they only establish that for any non-atomic game and its equilibrium, there exists at least one sequence of finite-player games converging to that game, with equilibria converging to the equilibrium. Motivated by this limitation, \cite{RefWorks:RefID:79-carmona2022approximation} characterize Nash equilibria in non-atomic games in terms of approximate equilibria in large finite-player games. Their results show that all sequences of finite-player games converging to a given non-atomic game have approximate equilibria that converge to a given equilibrium in the non-atomic game. The characterization result in this study shares the same motivation as that in \cite{RefWorks:RefID:79-carmona2022approximation}. In fact, this result can be seen as an extension of their characterization result to graphon games.

The most closely related work to this study is by \cite{Rokade2023}, which provides two key results: an equilibrium existence theorem under mild conditions and an approximation theorem for Nash equilibria. Like our study, \cite{Rokade2023} prove an equilibrium existence theorem under the assumptions of continuity and quasi-concavity of the utility functions. However, their equilibrium existence theorem addresses only the one-dimensional strategy space. In contrast, our existence theorem extends to infinite-dimensional cases. Moreover, our formulations and the techniques used in our proof are more familiar within the game theory literature, and our results do not require the strategy sets to be uniformly compact.

\cite{Rokade2023} demonstrate that if a sequence of Nash equilibria of sampled network games converges to a strategy profile of the corresponding graphon game as the number of players tends to infinity, this strategy profile is a Nash equilibrium of the graphon game with probability 1. However, their result leaves open the possibility that the limiting strategy profile may not be a Nash equilibrium of the graphon game if the sampling is biased.
Our study shows that when large network games—not necessarily sampled—converge to a graphon game, the limit of a sequence of asymptotic Nash equilibria in the large network games is \textit{always} a Nash equilibrium of the graphon game (see Corollary \ref{coro:limit is Nash:network}).

They also prove an approximate converse of their result. For any graphon game and its Nash equilibrium, they show that there exists at least one sequence of its sampled network games with approximate Nash equilibria that converge to the Nash equilibrium (see Theorem 3 and Remark 3 in \cite{Rokade2023}). In contrast, we prove that for any graphon game and its Nash equilibrium, \textit{every} sequence of network games converging to the graphon game has a sequence of asymptotic Nash equilibria that converges to the Nash equilibrium (see 1 \(\implies\) 3 in Theorem \ref{theorem:charactarization}).

In summary, the main difference between our study and that of \cite{Rokade2023} lies in the approaches used to study the relationship between graphon games and network games. While \cite{Rokade2023} adopt a stochastic interpretation of graphons and use statistical methods to analyze the relationship between graphon games and their sampled network games, our study takes a game-theoretic approach to analyze the relationship between graphon games and general network games. This game-theoretic approach allows us to characterize Nash equilibria in graphon games in relation to network games.

\section{Notations and definitions}
\label{sec:Notations and definitions}
\subsection{Preliminaries}\label{subsec: Mathematical preliminaries}
Let $(T,\Sigma,\mu)$ be a finite measure space and $E$ a Banach space. $(T,\Sigma,\mu)$ is \textit{essentially countably generated} if $\Sigma$ is generated by a countable subset of $\Sigma$ together with the $\mu$-null sets. Denote by $E^*$ the dual space of $E$, i.e. the space of bounded linear functionals from $E$ to $\mathbf{R}.$
A function $f:T \to E$ is \textit{strongly measurable} if there exists a sequence $\{\phi_n\}$ of simple functions such that $\mathrm{lim}_{n\to \infty} ||f(t) - \phi_n(t)|| = 0$ for a.e. $t \in T.$
The function $f$ is \textit{weakly measurable} if $x^* \circ f:T \to \mathbf{R}$ is measureble for all $x^* \in E^*.$ The function $f$ is \textit{essentially separably valued} if there exists a null set $N \in \Sigma$ such that $f(T \backslash N)$ is separable in $E$. Pettis' measurability theorem states that $f$ is strongly measurable if and only if $f$ is weakly measurable and essentially separably valued \citep[Theorem 2, p.42]{RefWorks:RefID:16-diestel1977vector}.
Therefore, when \(E\) is separable, the weak, strong, and usual measurability conditions are equivalent.


A strongly measurable function $f:T \to E$ is \textit{Bochner integrable} if there exists a sequence $\{\phi_n\}$ of simple functions such that 
\[
    \lim_{n\to \infty} \int_T ||f-\phi_n||d\mu =0.
\]
In this case, for each $S \in \Sigma$ the Bochner integral of $f$ over $S$ is defined by 
\[
\int_S f d \mu = \lim_{n\to \infty} \int_S \phi_nd\mu,
\]
    where the last limit is in the norm topology on $E.$
Denote by $L^1(\mu,E)$ the space of the equivalence classes of Bochner integrable functions $f:T\to E.$ The space $L^1(\mu, E)$ is a Banach space for the norm $||\cdot||_1,$ where  $||f||_1 = (\int_T ||f|| d\mu)$. Denote by $L^\infty(\mu,E)$ the space of essentially bounded functions, normed by the usual essential supremum norm $||\cdot||_\infty$.
We simplify the notation $L^1(\mu,\mathbf{R})$ and $L^\infty(\mu,\mathbf{R})$ to $L^1(\mu)$ and $L^\infty(\mu)$, respectively.

A mapping from $T$ to the family of subsets of $E$ is called a \textit{multifunction} or \textit{correspondence}. 
A multifunction $\Gamma : T \to E$ is said to be \textit{measurable} if the set $\{t \in  T| \Gamma(t) \cap U \not = \emptyset \}$ is in $\Sigma$ for every open subset $U$ of $E$. 
It is \textit{graph measurable} if its graph $ G_\Gamma = \{(t,x)\in T\times E| x \in \Gamma(t)\}$ belongs to $\Sigma \otimes \mathcal{B}(E)$, where $\mathcal{B}(E)$ is the Borel $\sigma$-algebra of $E$ generated by the norm topology. For nonempty closed valued multifunctions, measurability and graph measurability coincide whenever $(T,\Sigma, \mu)$ is complete and $E$ is separable. A function $f:T \to E$ is a \textit{selection} of $\Gamma$ if $f(t) \in \Gamma(t)$ for a.e. $t \in T.$ If $E$ is separable, then a nonempty multifunction $\Gamma$ with a measurable graph admits a measurable selection \citep[Theorem 1, p.54]{Hildenbrand1977}.

Let $B$ be a closed unit ball in $E$. A multifunction $\Gamma:T\to E$ is \textit{integrably bounded} if there exists $\phi \in L^1(\mu)$ such that $\Gamma(t) \subset \phi(t)B$ for a.e. $t \in T$. Denote by $\mathcal{S}^1_{\Gamma}$ the set of Bochner integrable selections of $\Gamma.$ If a nonempty multifunction $\Gamma$ is graph measurable and integrably bounded, then it admits a Bochner integrable selection whenever $E$ is separable.

\subsection{Network games}
Our focus is solely on network games \textit{with local aggregates}, which we will simply refer to as network games. In a network game, the network structure is described by an adjacency matrix $A = (A_{ij})_{i,j \in I}$, where $A_{ij}$ denotes the level of influence from agent $j$ to agent $i$. We allow for directed networks, which means that the matrix is not necessarily symmetric. When a strategy profile prevails, each agent aims to maximize their payoff, taking into account the externality defined by a local aggregate. The local aggregate is defined as the weighted sum of the strategy profile, with weights specified by the adjacency matrix $A$. 

The strategy space is a separable Banach space $E$.
\begin{definition}
    A network game $\mathcal{G}$ is a tuple $(I, (A_{i,j})_{i,j \in I}, (S_i)_{i \in I},(v_i)_{i \in I})$, where
    \begin{itemize}
        \item $I = \{1,2, \cdots n \}$ is the set of players;
        \item $(A_{ij})_{i,j \in I}$ is the adjacency matrix;
        \item $S_i \subset E$ is the strategy set of agent $i$;
        \item $v_i: S_i \times E \to \mathbf{R}$ is the utility function of agent $i$, and $v_i(a,e)$ is the utility of agent $i$ when he/she chooses a strategy $a \in S_i$ under a local aggregate $e \in E$.
        
    \end{itemize}
\end{definition}

\begin{definition}
    Let $\mathcal{G}$ be a network game.
        \begin{itemize}
            \item A strategy profile is an element of $\Pi_{i \in I} S_i.$
            \item Nash equilibrium for $\mathcal{G}$ is a strategy profile $s \in \Pi_{i \in I} S_i$ such that 
            $$v_i\left(s_i, \frac{1}{|I|}\sum_{j \in I} A_{ij} s_j \right) = \mathrm{max}_{a \in S_i} v_i\left(a ,\frac{1}{|I|}\sum_{j \in I} A_{ij} s_j \right) \text{ for all } i \in I.\footnote{$|A|$ denotes the cardinality of a set $A$.}$$
        \end{itemize}
    \end{definition}

\subsection{Graphon games}
\cite{parise2023} introduce graphon games, which extend network games to a continuum of agents. In a graphon game, the graphon quantifies the level of interaction between two agents in the population. We do not assume symmetry for the graphon. When a strategy profile prevails, each agent seeks to maximize their payoff by considering the externality described by a local aggregate. The local aggregate is defined as the weighted Bochner integral of the strategy profile, with the weights determined by the graphon.

To address network games and graphon games simultaneously, we define graphon games with a measure space of agents.
\begin{definition}
    A graphon game $\mathcal{G}$ with a measure space of agents is a tuple $((T, \Sigma, \mu), W, S, U)$, where
    \begin{itemize}
        \item $(T, \Sigma, \mu )$ is a complete, finite measure space of agents;
        \item $W$ is a \textit{graphon} on \(T\), that is, a measurable function from $T \times T$ to $[0,1]$;
        \item $S$ is a correspondence from $T$ to $E$, and $S(t)$ is the strategy set of agent $t$;
        \item $U: G_S \times E \to \mathbf{R}$ is a utility function,\footnote{\(G_S\) denotes the graph of the correspondence \(S\). See Section \ref{subsec: Mathematical preliminaries}} and $U(t,a,e)$ is the utility of agent $t$ when he/she chooses a strategy $a \in S(t)$ under a local aggregation $e \in E.$

    \end{itemize}
    When the Lebesgue interval serves as the space of agents, the game is referred to as a graphon game with a continuum of agents.
    \end{definition}

\begin{definition}
    Let $\mathcal{G}$ be a graphon game with a measure space of agents.
        \begin{itemize}
            \item A strategy profile is a Bochner integrable function $f: T \to E$ such that $f(t) \in S(t)$ for a.e $t \in T.$
            \item Nash equilibrium for $\mathcal{G}$ is a strategy profile $f$ such that 
            \[
                U(t, f(t), \int_T W(t,s) f(s) d \mu (s))
                = \mathrm{max}_{a \in S(t)} U(t, a ,\int_T W(t,s) f(s) d \mu (s)) \ \text{for a.e.}\  t \in T.
            \]  
        \end{itemize}
    \end{definition}
    \begin{remark}
        The integral $\int W(t,s) f(s) d \mu (s)$ is always defined for any strategy profile $f$ and $t \in T.$ In fact, Let $\phi (s) = W(t,s) f(s).$ For all $x^* \in E^*,$ the function $x^* \circ \phi$ is $\Sigma-\mathcal{B}(\mathbf{R})$ measurable because $f$ is weakly measurable and $x^* \circ \phi (s) = W(t,s) (x^* \circ f (s)).$ Since $E$ is separable, it follows from Pattis' measurability theorem that $\phi$ is strongly measurable. In addition, $\phi$ is Bohoner integlable because $\int ||\phi (s)|| d \mu (s) \leq \int ||f(s)|| d \mu(s) < \infty$
        \citep[see][Theorem 2, p.45]{RefWorks:RefID:16-diestel1977vector}.
    \end{remark}
    \begin{remark}
        Compared to the literature on large games, it is natural to use a separable Banach space as the infinite-dimensional strategy space in graphon games. In this literature, compact metric spaces are commonly used as the strategy space (see e.g., \cite{RefWorks:RefID:402-mas-colell1984theorem}, \cite{Hausman1988infinite}, and \cite{ RefWorks:RefID:122-carmona2020pure,RefWorks:RefID:394-carmona2021strict,RefWorks:RefID:79-carmona2022approximation}). Therefore, the space is a separable and complete metric space. In graphon games, a linear structure on the strategy space is necessary to calculate externalities. Hence, it is reasonable to adopt a separable Banach space for the strategy space and employ Bochner integration for calculating externalities. More general strategy spaces with weaker notions of integration, such as Pettis integration, are left for future research.
    \end{remark}
    
\section{Existence results}\label{sec: Existence results}
To establish the existence of pure-strategy equilibrium for graphon games, we consider the following assumptions. These assumptions are all standard in the game theory literature. 
\begin{enumerate}[\text{A.}1]
    \item\label{assum:essentially countably generated} \((T, \Sigma, \mu)\) is essentially countably generated.
    \item \label{assum:strategy} $S(t)$ is nonempty, convex and weakly compact for all $t.$
    \item The correspondence \(S\) is graph measurable and integrably bounded.
    \item\label{assum:continuity} $U(t, \cdot, \cdot): S(t) \times E \to \mathbf{R}$ is continuous for every $t \in T,$ where $S(t)$ and $E$ are endowed with the weak topology.
    \item \label{assum:quasi} $U(t, \cdot, e): S(t) \to \mathbf{R}$ is quasi-concave for every $t \in T$ and $e \in E$.
    \item \label{assum:measurability} $U(\cdot, \cdot, e):G_S \to \mathbf{R}$ is $\Sigma \otimes \mathcal{B}(E)- \mathcal{B}(\mathbf{R})$ measurable for every $e \in E$.
\end{enumerate}


We are now ready to provide the main result of this study.
\begin{theorem}\label{the:main}
        Let $\mathcal{G}$ be a graphon game with a measure space of agents. If $\mathcal{G}$ satisfies A.\ref{assum:essentially countably generated}-A.\ref{assum:measurability}, then it has a Nash equilibrium.
\end{theorem}
\begin{proof}
    See Section \ref{sec:proof of existence theorem}.
\end{proof}
As corollaries of Theorem \ref{the:main}, we can derive the equilibrium existence theorems for both network games and graphon games with a continuum of agents.
\begin{corollary}
Let $\mathcal{G} = (I, (A_{i,j})_{i,j \in I}, (S_i)_{i \in I},(v_i)_{i \in I})$ be a network game. Assume the following:
    \begin{enumerate}
        \item $S_i$ is nonempty, convex, and weakly compact for every $i \in I$.
        \item $v_i$ is continuous and $v_i(\cdot,e): A_i \to \mathbf{R}$ is quasi-concave for every $i \in I$.
    \end{enumerate}
    Then, $\mathcal{G}$ has a Nash equilibrium.
\end{corollary}

\begin{proof}
   Define a graphon game $\mathcal{G'} = ((T, \Sigma, \mu),W,S,U)$ with a measure space of agents as follows: 
\begin{itemize}
    \item $T = I,\  \Sigma = 2^T,$ and $\mu$ is the uniform probability measure on $(T, \Sigma).$
    \item $S(t) = S_t$ for all $t \in T.$
    \item $U(t, a, e) = v_t(a,e)$ for all $t\in T, a \in S(t)$, and $e \in E.$
    \item $W (t, s) = A_{ts}$ for all $t, s \in T.$ 
\end{itemize}
Then, it is straightforward to check A.\ref{assum:essentially countably generated}-A.\ref{assum:measurability} to hold. Thus, there exists a Nash equilibrium $s^*$ of $\mathcal{G'}.$ Since $\mu$ is the uniform probability measure on $(T, \Sigma),$ the strategy profile $s^*$ is also a Nash equilibrium of $\mathcal{G}.$ 
\end{proof}

\begin{corollary}\label{coro}
    Let $\mathcal{G}$ be a graphon game with a continuum of agents. If $\mathcal{G}$ satisfies A.\ref{assum:strategy}-A.\ref{assum:measurability}, then it has a Nash equilibrium.
\end{corollary}
\begin{proof}
    Since the Lebesgue interval is complete, finite, and essentially countably generated, the result immediately follows from Theorem \ref{the:main}.
\end{proof}
\cite{parise2023} impose strict concavity on the utility function, resulting in each agent's best response being single-valued.\footnote{\cite{Rokade2023} also point out this aspect.} Additionally, they assume a condition on the graphon's eigenvalues to make the best response a contraction, mainly focusing on graphon games with a unique equilibrium. In contrast, our assumptions allow for multi-valued best responses and multiple equilibria. We provide an example of graphon games with multi-valued best responses and multiple equilibria. In our example, we consider a graphon game with a variant of linear quadratic utility functions. When the network effect via the graphon is suitably small, the set of Nash equilibria has the cardinality of the continuum. We illustrate that some Nash equilibria can be directly constructed from the graphon.
\begin{example}
Consider the following measure space of agents, strategy space, strategy set, and utility function.
\begin{itemize}
\item The measure space is the Lebesgue Interval $(T, \Sigma, \mu).$
\item The strategy space is $\mathbf{R},$ and the strategy set of for each agent $t$ is $S(t) = [0,L] \ (L>0).$
\item Let $\lambda \geq 0$. The common utility function $u (a, e)$ is
$$u(a,e) = 
\begin{cases}
    -\frac{1}{2} a^2 + \lambda e a & (a < \lambda e)\\
    \frac{1}{2} \lambda^2 e^2 & (\lambda e \leq a \leq \lambda e + 1) \\
    -\frac{1}{2} (a-1)^2 + \lambda e (a-1) & (\lambda e + 1 < a).
\end{cases}$$
\end{itemize}
Then, the best response of agent $t$ under a local aggregate $e$ is multi-valued as follows:
$$B(t,e) = 
\begin{cases}
\{0 \} & \text{if } \lambda e + 1 < 0\\
\{L\}  & \text{if } L < \lambda e \\
[\lambda e, \lambda e + 1]\cap S(t) & \text{otherwise}
.\end{cases}
$$

For any graphon $W,$ it follows from Theorem \ref{the:main} that there exists a Nash equilibrium. Furthermore, if $W$ satisfies $ \lambda \|W\|_{\infty} < 1$, then there exist multiple equilibria. Specifically, the following proposition holds. Denote by $\mathcal{S}$ the set of all measurable functions $g:T \to \mathbf{R}$ such that $g(t) \in [0,1]$ for a.e. $t \in T.$ For a given graphon $W,$ we define $W_n : T \times T \to [0,1] \ (n =1, 2, \cdots )$ by induction as follows:
\begin{align*}
     W_1(t,s) &= W(t,s); \\
     W_n(t,s) &= \int_T W(t,x) W_{n-1}(x, s) dx\ \text{ for } n \geq 2.\\
\end{align*}

\begin{proposition}\label{proposition}
    Let $\mathcal{G} = ((T, \Sigma, \mu), W, S, U)$ be a graphon game, where $(T, \Sigma, \mu), S$ and $U$ are defined above. Suppose $\lambda || W ||_\infty < 1$ and $L$ is sufficiently large.\footnote{\(\lambda \|W\|_{\infty}\) represents how much each agent's utility is influenced by the choices of other agents through the network. As mentioned in the proof, we assume that \(L\) is sufficiently large such that
$\frac{1}{1-\lambda||W||_\infty} \leq L$
and
$\frac{\lambda }{1-\lambda||W||_\infty} + 1 \leq L.$
}
    Then, for all $g \in \mathcal{S}$, the measurable function $s_g \in L^1(\mu)$ defined as follows is a Nash equilibrium of $\mathcal{G}:$
    $$s_g(t) = g(t) + \lambda \int_T \Gamma (t,s,\lambda) g(s) dy, $$
    where $\Gamma (t,s,\lambda) = \sum_{n=1}^\infty \lambda^{n-1} W_n(t,s).$
 In addition the map $g \mapsto s_g$ is an injection, that is, if $g_1$ and $g_2$ differ as elements of $\mathcal{S},$ then $s_{g_1}$ and $s_{g_2}$ are different Nash equilibria. 
\end{proposition}

Denote by $NE(\mathcal{G})$ the set of all Nash equilibria of $\mathcal{G}.$ The latter part of Proposition \ref{proposition} implies that the cardinality of $NE(\mathcal{G})$ is  $c$, the cardinality of the continuum. This is because $|\mathcal{S}| \leq |NE (\mathcal{G})| \leq |L^1(\mu)| \text{ and } |\mathcal{S}| = |L^1(\mu)| = c$.

\begin{proof}
    See Section \ref{sec: proof of example 2}.
\end{proof}

For example, let \( W(t,s) = t^\alpha s^{1-\alpha} \) with \( \alpha \in (0,1) \) and \( \lambda < 1 \).
From simple calculations, we know that 
\[ W_n(t,s) = \frac{1}{2^{n-1}} t^\alpha s^{1-\alpha}.\]
Therefore, we have
\[
\Gamma(t,s, \lambda) = \sum_{n=1}^\infty \lambda^{n-1} W_n(t,s) = \sum_{n=1}^\infty \lambda^{n-1} \left( \frac{1}{2^{n-1}} t^\alpha s^{1-\alpha} \right)
= t^\alpha s^{1-\alpha} \sum_{n=1}^\infty \left( \frac{\lambda}{2} \right)^{n-1}
= \frac{2}{2-\lambda} t^\alpha s^{1-\alpha}.
\]
Assuming \( g(t) = 1 \) for all $t \in [0,1]$, we obtain a Nash equilibrium \(s_g\) as

\[
s_g(t) = g(t) + \lambda \int_0^1 \Gamma(t,s, \lambda) g(s) \, ds
= 1 + \lambda \int_0^1 \frac{2}{2-\lambda} t^\alpha s^{1-\alpha} \, ds
= 1 + \frac{2\lambda}{2-\lambda} t^\alpha \int_0^1 s^{1-\alpha} \, ds
= 1 + \frac{2\lambda}{(2-\lambda)(2-\alpha)} t^\alpha.
\]

\end{example}

\section{Characterization of Nash equilibria of graphon games}
\label{sec:Characterization}
\subsection{Setup}
In this section, we consider graphon games and network games where all players have the same strategy set.
The formal setup is as follows.
The strategy set \( S \) is common to all players and is a weakly compact convex subset of a separable Banach space that contains 0.
Let \( C \) be the space of real-valued functions on \( S^2 \) that are continuous with respect to the weak topology of \(S\), endowed with the supremum norm \(\| \cdot \|\). 
We restrict the set of utility functions to the subset \( C_0 \) of \( C \), consisting of functions whose sup norm is less than or equal to 1. Since any function \( f \in C \) is bounded, this restriction imposes no constraints in terms of preference relations.
As \( S \) is  weakly compact, and the weak topology of \( S \) is metrizable \citep[][Theorem 3, p.434]{dunford1988linear}, the Banach space \( C \) is separable \citep[][Lemma 3.99]{RefWorks:RefID:7-aliprantis2006infinite}.

For notational simplicity, we consider the space of agents in a graphon game as \( (0,1] \).
We assume that the utility function profile is measurable.
Thus, in this section, we focus on a graphon game defined by \( ((T, \Sigma, \mu), W, U, S) \), where \( T = (0,1] \), \( \Sigma \) is the \(\sigma\)-algebra consisting of Lebesgue measurable subsets of \( (0,1] \), \( \mu \) is the Lebesgue measure, and \( U \) is a measurable function from \( (0,1] \) to \( C_0 \).\footnote{\(U(t) \in C_0\) represents the utility function of agent \(t\). Note that in Section \ref{sec: Existence results}, \(U\) is a function from \(G_S \times E\) to \(\mathbf{R}\), while in this section, \(U\) is a function from \((0,1]\) to \(C_0\).}
We refer to the graphon game simply as \( (W, U) \), specifying only the graphon and utility functions.
Similarly, as for network games, we consider only a network game defined by \( (I, (A_{i,j})_{i,j\in I}, (v_i)_{i \in I}, (S_i)_{i \in I}) \), where \( S_i = S \) for all \( i \in I \), and we represent it as \( (I, (A(i,j))_{i,j\in I}, (v(i))_{i \in I}) \).

Let \( f : (0,1] \to S \) be a measurable function. Then, since \( S \) is closed convex and contain 0, it follows that
\(
\int W(t,s) f(s) \, d\mu(s) \in S
\)
for all \(t \in (0,1]\)
\citep[][Corollary 8, p.48]{RefWorks:RefID:16-diestel1977vector}.
For any \( \epsilon > 0 \), the set
\[
\{ t \in (0,1] : U(t)(f(t), \int W(t,s) f(s) \, d\mu(s)) \geq \max_{a \in S} U(t)(a, \int W(t,s) f(s) \, d\mu(s)) - \epsilon \}
\]
is measurable. In fact, if we define \( \mathbb{W}f(t) = \int W(t,s) f(s) \, d\mu(s) \) for all \( t \in (0,1] \), then this set coincides with
\[
(U, f, \mathbb{W}f)^{-1}(\{(u,a,e) \in C \times S \times S : u(a,e) \geq \max_{a' \in S} u(a',e) - \epsilon\}).
\]
This set is measurable because
\[
\{(u,a,e) \in C \times S \times S : u(a,e) \geq \max_{a' \in S} u(a',e) - \epsilon\}
\]
is closed in \( C \times S \times S \), and \( (U, f, \mathbb{W}f) \) is jointly measurable.\footnote{As for the measurability of \( \mathbb{W}f \), see Lemma \ref{proposition: measurability}
in the Appendix and the subsequent argument.}

Thus, we can define a concept of approximate equilibrium for both network games and graphon games as follows.
\begin{definition}
    Let \( \epsilon \) be a nonnegative number.
\begin{enumerate}
    \item Let \( \mathcal{G} = (I, (A(i,j))_{i,j\in I}, (v(i))_{i \in I}) \) be a network game. A strategy profile \( s \) of \( \mathcal{G} \) is an \( \epsilon \)-\textit{Nash equilibrium} if
    \[
    \frac{| \{ i \in I : v(i)(s(i), \sum_{j \in I} A(i,j) s(j)) \geq \max_{a \in S} v(i)(a, \sum_{j \in I} A(i,j) s(j)) - \epsilon \} |}{|I|} \geq 1 - \epsilon.
    \]
    \item Let \( \mathcal{G} = (W, U) \) be a graphon game. A strategy profile \( f \) of \( \mathcal{G} \) is an \( \epsilon \)-\textit{Nash equilibrium} if
    \[
    \mu \left( \left\{ t \in (0,1] : U(t)(f(t), \int W(t,s) f(s) d\mu(s)) \geq \max_{a \in S} U(t)(a, \int W(t,s) f(s) d\mu(s)) - \epsilon \right\} \right) \geq 1 - \epsilon.
    \]
\end{enumerate}
\end{definition} 
We say a sequence \(\{\mathcal{G}_n = (W_n,U_n)\}\) of graphon games \textit{converge to} a graphon game \(\mathcal{G} = (W,U)\), denoted by \( \mathcal{G}_n \to \mathcal{G} \), if \(W_n \overset{a.e.}{\longrightarrow} W\) and \(U_n \overset{a.e.}{\longrightarrow} U\).\footnote{\(W_n \overset{a.e.}{\longrightarrow} W\) means \(W_n\) converges almost everywhere to \(W\). Similarly, \(U_n \overset{a.e.}{\longrightarrow} U\) means \(U_n\) converges almost everywhere to \(U\).
} 
Finally, for a network game \(( I, (A(i,j))_{i,j\in I}, (v(i))_{i \in I} )\) with \( I = \{1, 2, \dots, n\} \) and a strategy profile \(s\) of \(\mathcal{G}\), we define the \( n \)-step graphon \( W_A \), the \( n \)-step function \( U_v : (0,1] \to C_0 \), and the \( n \)-step function \( f_s : (0,1] \to S \) as follows:\footnote{A graphon \(W:(0,1]^2 \to [0,1]\) is called an \(n\)-\textit{step graphon} if it is constant on each rectangle \(\left(\frac{i-1}{n}, \frac{i}{n}\right] \times \left(\frac{j-1}{n}, \frac{j}{n}\right] \) for all \(i,j = 1, 2, \dots, n\).
Similarly, we call a function \(f\) from \( (0,1]\) to a set \(X\) an \(n\)-\textit{step function} if it is constant on each interval \(\left(\frac{i-1}{n}, \frac{i}{n}\right]\) for \(i = 1, 2, \dots, n\).
}
\[
W_A(t,s) = \sum_{i,j=1}^n A(i,j) \chi_{\left(\frac{i-1}{n}, \frac{i}{n}\right]}(t) \chi_{\left(\frac{j-1}{n}, \frac{j}{n}\right]}(s);
\]
\[
U_v(t) = \sum_{i=1}^n v(i) \chi_{\left(\frac{i-1}{n}, \frac{i}{n}\right]}(t);
\]
\[
f_s(t) = \sum_{i=1}^n s(i) \chi_{\left(\frac{i-1}{n}, \frac{i}{n}\right]}(t).
\]

\subsection{Characterization results}
The following proposition is a consequence of the Lebesgue differentiation theorem \citep[Theorem 7.10]{RefWorks:RefID:13-rudin1987real} and, in our context, states that any graphon can be approximated by a sequence of networks.
\begin{lemma}[\cite{RefWorks:RefID:410-borgs2008convergent}]
\label{lem:approximation graphon}
    For any graphon $W$, there exists a sequence $\{W_n\}$ of graphons such that $W_n$ is an $n$-step graphon for all $n \in \mathbf{N}$, and $W_n \overset{a.e.}{\longrightarrow} W.$
    \footnote{See Lemma 3.2 of \cite{RefWorks:RefID:410-borgs2008convergent}.}
\end{lemma}

By extending the Lebesgue differentiation theorem to Bochner integration, we can derive the following proposition, which allows us to approximate the utility functions and strategy profile of a graphon game by those of network games.
\begin{lemma}\label{lem:approximation strategy}
    Let $E$ be a Banach space, $C$ be a closed convex subset of $E$, $f:(0,1] \to C$ be Bochner integrable. Then, for any sequence \(\{k_n\}\) of integers with \(k_n \to \infty\), there exists a sequence $\{f_n\}$ of functions such that $f_n$ is a $k_n$-step function from $(0,1]$ to $C$, and $f_n \overset{a.e.}{\longrightarrow} f.$ 
\end{lemma}

\begin{proof}
See Section \ref{sec:Proof of Lemma: approximation strategy}.    
\end{proof}

By applying these results, any graphon game can be regarded as the limit of a sequence of large network games:
\begin{lemma}\label{lem:approximation graphon game}
    For any graphon game \(\mathcal{G} = (W, U)\), there exists a sequence \(\{\mathcal{G}_n = ( I_n, (A_n(i,j))_{i,j\in I_n}, (v_n(i))_{i \in I_n} ) \}\) of network games such that \( W_{A_n} \overset{a.e.}{\longrightarrow} W \), \( U_{v_n} \overset{a.e.}{\longrightarrow} U \), and \( |I_n| \to \infty \).
Furthermore, for any strategy profile \(f\) of \(\mathcal{G}\) and any sequence \(\{\mathcal{G}_n\}\) of network games that satisfies these properties, there exists a sequence \(\{s_n\}\) of strategy profiles such that \(s_n\) is a strategy profile of \(\mathcal{G}_n\) for all $n \in \mathbf{N},$ and \(f_{s_n} \overset{a.e.}{\longrightarrow} f\).

\end{lemma}

\begin{proof}
    Note that $U:(0,1] \to C_0$ is measurable by definition, and integrably bounded because $||U(t)||\leq 1$ for all $t\in (0,1].$ Thus, \(U\) is Bochner integrable.
    Since $C$ is a Banach space and $C_0$ is a closed convex subset of $C$, it follows from Lemma \ref{lem:approximation strategy} that there exists a sequence $\{U_n\}$ such that $U_n$ is a $n$-step function from $(0,1]$ to $C_0$ and $U_n \overset{a.e.}{\longrightarrow} U.$
    
    Similarly, it follows from Lemma \ref{lem:approximation graphon} that there exists a sequence $\{W_n\}$ of graphons such that $W_n$ is a $n$-step graphon and $W_n \overset{a.e.}{\longrightarrow} W.$ It is clear that for each \(n \in \mathbf{N}\), there exists a network game \(\mathcal{G}_n = ( I_n, (A_n(i,j))_{i,j\in I_n}, (v_n(i))_{i \in I_n} ) \) such that $|I_n| = n,$ $W_{A_n} = W_n$, and $U_{v_n} = U_n.$ The sequence $\{\mathcal{G}_n\}$ of network games satisfies the required properties. 
    
    As for the latter part, let $f$ be a strategy profile of $\mathcal{G}$ and \( \{ \mathcal{G}_n = ( I_n, (A_n(i,j))_{i,j\in I_n}, (v_n(i))_{i \in I_n} ) \} \) be a sequence of network games such that \( W_{A_n} \overset{a.e.}{\longrightarrow} W \),
    \( U_{v_n} \overset{a.e.}{\longrightarrow} U \),
    and \( |I_n| \to \infty \). Let \( k_n = |I_n| \). Then, we have \( k_n \to \infty \). Therefore, by Lemma \ref{lem:approximation strategy}, there exists a sequence $\{f_n\}$ of functions such that $f_n$ is a $k_n$-step function from $(0,1]$ to $S,$ and $f_n \overset{a.e.}{\longrightarrow} f$. Again, it is clear that there exists a strategy profile \( s_n \) on \( \mathcal{G}_n \) such that \( f_{s_n} = f_n \). The sequence \( \{s_n\} \) of strategy profiles satisfies the required properties.

\end{proof}

We say that a sequence of network games converges to a graphon game if it satisfies the properties in Lemma \ref{lem:approximation graphon game}:
\begin{definition}
    A sequence of network games \( \{\mathcal{G}_n = ( I_n, (A_n(i,j))_{i,j\in I_n}, (v_n(i))_{i \in I_n} ) \)\} \textit{converges to} a graphon game \( \mathcal{G} \), denoted by \( \mathcal{G}_n \to \mathcal{G} \), if it holds that \( W_{A_n} \overset{a.e.}{\longrightarrow} W \), \( U_{v_n} \overset{a.e.}{\longrightarrow} U \), and \( |I_n| \to \infty \).
\end{definition}
Then, Lemma \ref{lem:approximation graphon game} states that for any graphon game, there exists a sequence of network games that converges to the graphon game.

\begin{remark}
    The convergence of networks to graphons considered here differs from that in the recent theory of graph limits. In graph limit theory, graphons were introduced as a more refined limit concept of graphs, described by the \textit{cut metric}. However, this limit concept is difficult to handle with utility functions and strategies. In fact, a network game \(( I, (A(i,j))_{i,j\in I}, (v(i))_{i \in I} )\) can be regarded as a \textit{decorated graph}, where $I$ is the set of nodes, $A$ is a weighted graph on $I$, and each node \(i\in I\) is assigned a utility function \(v(i)\). It is challenging to give game-theoretic interpretation to the limiting objects of such decorated graphs (see \cite{RefWorks:RefID:416-kunszenti-kovács2019uniqueness} and \cite{RefWorks:RefID:415-kunszenti-kovács2022multigraph}). Therefore, we adopt an approach that continuously represents network games as graphon games, similar to the approach adopted by \cite{RefWorks:RefID:265-kannai1970continuity} in general equilibrium theory.
\end{remark}

We first prove that if a sequence of graphon games converges to a graphon game, then any convergent sequence of asymptotic Nash equilibria converges to a Nash equilibrium of the limiting graphon game.
\begin{theorem}\label{theorem:limit is Nash}
Let \( \mathcal{G} \) be a graphon game and \( f \) be a strategy profile of \( \mathcal{G} \).  
Suppose sequences \( \{ \mathcal{G}_n \} \), \( \{ f_n \} \), and \( \{ \epsilon_n \} \) of graphon games, strategy profiles, and nonnegative numbers satisfy the following properties:
\begin{itemize}
    \item \( f_n \) is an \( \epsilon_n \)-Nash equilibrium of \( \mathcal{G}_n \) for all $n \in \mathbf{N}$,
    \item \( \mathcal{G}_n \to \mathcal{G} \), \( f_n \overset{a.e.}{\longrightarrow} f \), and \( \epsilon_n \to 0 \).
\end{itemize}
Then, \( f \) is a Nash equilibrium of \( \mathcal{G} \).
\end{theorem}

\begin{proof}
See Section \ref{sec: Proof of limit is Nash}.     
\end{proof}

As a corollary of this theorem, the same result holds for the limits of large network games.
\begin{corollary}\label{coro:limit is Nash:network}
    Let \( \mathcal{G} \) be a graphon game and \( f \) be a strategy profile.  
    Suppose sequences \( \{\mathcal{G}_n\} \), \( \{s_n\} \), and \( \{\epsilon_n\} \) of network games, strategy profiles, and nonnegative numbers satisfy the following properties:
    \begin{itemize}
        \item \( s_n \) is an \( \epsilon_n \)-Nash equilibrium of \( \mathcal{G}_n \) for all $n \in \mathbf{N}$,
        \item \( \mathcal{G}_n \to \mathcal{G} \), \( f_{s_n} \overset{a.e.}{\longrightarrow} f \), and \( \epsilon_n \to 0 \).
    \end{itemize}
    Then, \( f \) is a Nash equilibrium of \( \mathcal{G} \).
\end{corollary}

\begin{proof}
Let \( \mathcal{G}_n = ( I_n, (A_n(i,j))_{i,j\in I_n}, (v_n(i))_{i \in I_n} )\).
Define a sequence  \( \{\mathcal{G}'_n = (W_n, U_n)\}\) of graphon games as \((W_n, U_n) = (W_{A_n}, U_{v_n})\) for all \(n \in \mathbf{N}\), and define a strategy profile \( f_n \) of \( \mathcal{G}'_n \) as \( f_n = f_{s_n} \).  
Then,
\[
\frac{| \{ i \in I_n : v_n(i)(s_n(i), \sum_{j \in I_n} A_n(i,j) s_n(j) \geq \max_{a \in S} v_n(i)(a, \sum_{j \in I_n} A_n(i,j) s_n(j)) - \epsilon_n \} |}{|I_n|}  
\]
is equal to
\[
\mu \left( \left\{ t \in (0,1] : U_n(t)(f_n(t), \int W_n(t,s) f_n(s) d\mu(s)) \geq \max_{a \in S} U_n(t)(a, \int W_n(t,s) f_n(s) d\mu(s)) - \epsilon_n \right\} \right).
\]
Since \( s_n \) is an \( \epsilon_n \)-Nash equilibrium of \( \mathcal{G}_n \), it follows that \( f_n \) is an \( \epsilon_n \)-Nash equilibrium of \( \mathcal{G}'_n \).  
Since \( \mathcal{G}'_n \to \mathcal{G} \) and \( f_n \overset{a.e.}{\longrightarrow} f \) by assumption, it follows from Theorem \ref{theorem:limit is Nash} that \( f \) is a Nash equilibrium of \( \mathcal{G} \).

\end{proof}
Next, we prove the converse of Theorem \ref{theorem:limit is Nash}: If a sequence of graphon games converges to a graphon game, it is only asymptotic Nash equilibria that converge to a Nash equilibrium of the limiting graphon game.

\begin{theorem}\label{theorem:Nash is limit}
Let \( \mathcal{G} \) be a graphon game and \( f \) be a Nash equilibrium of \( \mathcal{G} \).  
Consider sequences \( \{ \mathcal{G}_n \} \) and \( \{ f_n \} \) of graphon games and strategy profiles such that \( \mathcal{G}_n \to \mathcal{G} \) and \( f_n \overset{a.e.}{\longrightarrow} f \).
Then, there exists a sequence \( \{\epsilon_n\} \) of nonnegative numbers such that \( f_n \) is an \( \epsilon_n \)-Nash equilibrium of \( \mathcal{G}_n \) for all $n \in \mathbf{N}$, and \( \epsilon_n \to 0 \).

\end{theorem}

\begin{proof}
See Section \ref{sec:Proof of Nash is limit}.    
\end{proof}

Again, the same result holds for sequences of large network games.
\begin{corollary}\label{coro:Nash is limit:network}
    Let \( \mathcal{G} \) be a graphon game and \( f \) be a Nash equilibrium of \( \mathcal{G} \).  
    Consider sequences \( \{\mathcal{G}_n\} \) and \( \{s_n\} \) of network games and strategy profiles such that \(s_n\) is a strategy profile of \(\mathcal{G}_n\) for all \(n \in \mathbf{N}\), \( \mathcal{G}_n \to \mathcal{G} \), and
    \( f_{s_n} \overset{a.e.}{\longrightarrow} f \).
    Then, there exists a sequence \( \{\epsilon_n\} \) of nonnegative numbers such that \( s_n \) is an \( \epsilon_n \)-Nash equilibrium of \( \mathcal{G}_n \) for all $n \in \mathbf{N}$, and \( \epsilon_n \to 0 \).
\end{corollary}

\begin{proof}
Let \( \mathcal{G}_n = ( I_n, (A_n(i,j))_{i,j\in I_n}, (v_n(i))_{i \in I_n} )\).
Define a sequence  \( \{\mathcal{G}'_n = (W_n, U_n)\}\) of graphon games as \((W_n, U_n) = (W_{A_n}, U_{v_n})\) for all \(n \in \mathbf{N}\), and define a strategy profile \( f_n \) of \( \mathcal{G}'_n \) as \( f_n = f_{s_n} \).
By assumption, we have \( \mathcal{G}'_n \to \mathcal{G} \) and \( f_n \overset{a.e.}{\longrightarrow} f \).
Therefore, by Theorem \ref{theorem:Nash is limit}, there exists a sequence \(\{ \epsilon_n \}\) of nonnegative numbers such that \( f_n \) is an \( \epsilon_n \)-Nash equilibrium of \( \mathcal{G}'_n \) for all $n \in \mathbf{N},$ and \( \epsilon_n \to 0 \).  

As in the proof of Corollary \ref{coro:limit is Nash:network}, 
\[
\frac{| \{ i \in I_n : v_n(i)(s_n(i), \sum_{j \in I_n} A_n(i,j) s_n(j)) \geq \max_{a \in S} v_n(i)(a, \sum_{j \in I_n} A_n(i,j) s_n(j)) - \epsilon_n \} |}{|I_n|} 
\]
is equal to
\[
\mu \left( \left\{ t \in (0,1] : U_n(t)(f_n(t), \int W_n(t,s) f_n(s) d\mu(s)) \geq \max_{a \in S} U_n(t)(a, \int W_n(t,s) f_n(s) d\mu(s)) - \epsilon_n \right\} \right).
\] 
Therefore, \( s_n \) is an \( \epsilon_n \)-Nash equilibrium of \( \mathcal{G}_n \).

\end{proof}
Summarizing the results so far, we obtain the following theorem. Intuitively, this theorem states that a strategy profile of a graphon game is a Nash equilibrium if and only if it is the limit of a sequence of asymptotic Nash equilibria in large network games (see \(1 \iff 2\) in this theorem). It also states that for any graphon game and its Nash equilibrium, every sequence of network games converging to the graphon game has a corresponding sequence of asymptotic Nash equilibria that converges to the Nash equilibrium (see \(1 \implies 3\) in this theorem).
\begin{theorem}\label{theorem:charactarization}
    Let \( \mathcal{G} \) be a graphon game and \( f \) be a strategy profile. Then, the following conditions are equivalent:
    \begin{enumerate}
        \item \( f \) is a Nash equilibrium of \( \mathcal{G} \).
        \item There exist sequences \( \{\mathcal{G}_n\} \), \( \{s_n\} \), and \( \{\epsilon_n\} \) of network games, strategy profiles, and nonnegative numbers that satisfy the following properties:
        \begin{itemize}
            \item \( s_n \) is an \( \epsilon_n \)-Nash equilibrium of \( \mathcal{G}_n \) for all \( n \),
            \item \( \mathcal{G}_n \to \mathcal{G} \), \( f_{s_n} \overset{a.e.}{\longrightarrow} f \), and \( \epsilon_n \to 0 \).
        \end{itemize}
        \item For any sequence \( \{\mathcal{G}_n\} \) of network games such that \( \mathcal{G}_n \to \mathcal{G} \), there exist sequences \( \{s_n\} \) and \( \{\epsilon_n\} \) of strategy profiles and nonnegative numbers that satisfy the following properties:
        \begin{itemize}
            \item \( s_n \) is an \( \epsilon_n \)-Nash equilibrium of \( \mathcal{G}_n \) for all \( n \),
            \item \( f_{s_n} \overset{a.e.}{\longrightarrow} f \) and \( \epsilon_n \to 0 \).
        \end{itemize}
    \end{enumerate}
\end{theorem}

\begin{proof}
2\(\implies\)1: This is Corollary \ref{coro:Nash is limit:network}.

1\(\implies\)3:  
Let \( \{ \mathcal{G}_n \} \) be a sequence of network games such that \( \mathcal{G}_n \to \mathcal{G} \). By the latter part of Lemma \ref{lem:approximation graphon game}, there exists a sequence \( \{ s_n \} \) of strategy profiles such that \( s_n \) is a strategy profile of \( \mathcal{G}_n \) for all \( n \in \mathbf{N} \), and \( f_{s_n} \overset{a.e.}{\longrightarrow} f \). Since \( f \) is a Nash equilibrium of \( \mathcal{G} \), it follows from Corollary \ref{coro:Nash is limit:network} that there exists a sequence \( \{ \epsilon_n \} \) of nonnegative numbers such that \(  s_n \) is an \( \epsilon_n \)-Nash equilibrium of \( \mathcal{G}_n \) for all \( n \in \mathbf{N} \), and \( \epsilon_n \to 0 \). The sequences \( \{ s_n \} \) and \( \{ \epsilon_n \} \) satisfy the required properties.

3\(\implies\)2:  
For the graphon game \( \mathcal{G} \), it follows from Lemma \ref{lem:approximation graphon game} that there exists a sequence \( \{ \mathcal{G}_n \} \) of network games such that \( \mathcal{G}_n \to \mathcal{G} \). By condition 3, there exist sequences \( \{ s_n \} \) and \( \{ \epsilon_n \} \) of strategy profiles and nonnegative numbers, respectively, such that \( s_n \) is an \( \epsilon_n \)-Nash equilibrium of \( \mathcal{G}_n \) for all \( n \in \mathbf{N} \), \( f_{s_n} \overset{a.e.}{\longrightarrow} f \), and \( \epsilon_n \to 0 \). The sequences \( \{ \mathcal{G}_n \} \), \( \{ s_n \} \), and \( \{ \epsilon_n \} \) satisfy the required properties.

\end{proof}

\section{Concluding remarks}
\label{sec:concluding remarks}
Two important questions remain for future research. The first concerns the conditions under which the convexity of players' strategy sets and the quasi-concavity of their utility functions can be removed in the equilibrium existence result. In the case of the complete graphon, i.e., \( W(t, s) \equiv 1 \), each player's payoff depends on their own strategy and the average strategy of others. In this case, under certain conditions on the space of agents (see \cite{RefWorks:RefID:423-khan1999non-cooperative}, \cite{RefWorks:RefID:274-sun2015pure-strategy}, and \cite{RefWorks:RefID:70-he2022conditional}), these convexity and quasi-concavity assumptions can be relaxed. Further investigation is required to identify the types of graphons for which these conditions can be removed (see also Proposition 1' in \cite{Rokade2023}). The second question is whether Nash equilibria in graphon games can be characterized as the limits of Nash equilibria in network games. While it is known from Corollary \ref{coro:limit is Nash:network} that the limit of a sequence of Nash equilibria in network games is a Nash equilibrium in the graphon game, it remains uncertain whether for any graphon game and its Nash equilibrium, there exists a sequence of network games that converges to the graphon game and has Nash equilibria converging to the Nash equilibrium of the graphon game.
\section{Appendix}

\subsection{Proof of Theorem \ref{the:main}}\label{sec:proof of existence theorem}
 Define $B(t, e) = \mathrm{argmax}_{a \in S(t)} U(t,a,e).$
\begin{lemma}
    For every $t \in T,$ the correspondence $B(t, \cdot): E \twoheadrightarrow S(t)$ is nonempty, convex valued, and upper hemicontinuous in the weak topology of $E.$
\end{lemma}
\begin{proof}
    The result immediately follows from A.\ref{assum:strategy}-A.\ref{assum:quasi}.
\end{proof}

\begin{lemma}\label{proposition: measurability}
Let $(T, \Sigma^T, \mu), (S, \Sigma^S, \nu)$ be finite complete measure spaces and $f:T \times S \to E$ a Bochner integrable function. Then, there exists a measurable set $T' \in \Sigma^T$ with $\mu (T \setminus T') = 0$ such that
    \begin{itemize}
        \item for each $t \in T',$ $f_t:S \to E \ (s \mapsto f(t,s) )$ is Bochner integrable;
        \item $I_T (t) = 
        \begin{cases}
        \int_S f_t(s) d\nu(s) & t \in T' \\
        0 & otherwise
        \end{cases}
        $ 
        is $\Sigma^T-\mathcal{B}(E)$ measurable.\\
    \end{itemize}
\end{lemma}
\begin{proof}
Since $f$ is essentially separably valued, we may assume that $E$ is separable. For all $t \in T$, since $f_t$ is measurable, it is weakly measurable, and thus strongly measurable by Pettis's measurability theorem.

Since $f$ is Bochner integrable, $\|f\|$ is Lebesgue integrable. It follows from Fubini’s theorem for Lebesgue integration that
\begin{equation*}
    \int_S \|f_t(s)\| d\nu(s) < \infty
\end{equation*}
for a.e. $t \in T$.
Thus, there exists a measurable set $T' \in \Sigma^T$ with $\mu(T \setminus T') = 0$ such that $f_t$ is Bochner integrable for all $t \in T'$. Define 
$I_T (t) = \int_S f_t(s) d\nu (s)$ if $t \in T'$ and $0$ otherwise.

To prove the measurability of $I_T$, it is sufficient to prove its weak measurability. Take $x^* \in E^*$ arbitrarily. Since $x^*$ is a bounded operator and $f$ is Bochner integrable, $x^* \circ f$ is Lebesgue integrable. Hence, again by Fubini’s theorem, we can define $g_{x^*} (t) = \int_S (x^* \circ f)_t(s) d\nu(s)$ for a.e. $t \in T$, and $g_{x^*}$ is measurable. Meanwhile, note that $(x^* \circ I_T)(t) = \int_S (x^* \circ f_t)(s) d \nu(s)$ for all $t \in T'$. Therefore, $g_{x^*}$ and $x^* \circ I_T$ coincide outside a null set. Since $X$ is complete, the measurability of $g_{x^*}$ implies that of $x^* \circ I_T$.
\end{proof}

Define $\mathbb{W}:T \times L^1 (\mu,E) \to E$ as $\mathbb{W} (t, f) = \int_T W(t,s) f(s) d \mu (s).$
Then, $\mathbb{W}(t, \cdot): L^1 (\mu,E) \to E$ is continuous in the respective weak topologies for every $t \in T.$
For every $f \in L^1(\mu,E)$, since $W(t, s) f(s)$ is Bohoner integrable on $T \times T,$ and since $(T, \Sigma, \mu)$ is complete, it follows from Lemma \ref{proposition: measurability} that $\mathbb{W} (\cdot, f): T \to E$ is $\Sigma - \mathcal{B}(E)$ measurable.

\begin{lemma}
        $B(\cdot, \mathbb{W}(\cdot, f)): T \twoheadrightarrow E$ has a measurable graph for every $f \in L^1(\mu,E).$
\end{lemma}

\begin{proof}
    Since $U: G_{S} \times E \to E$ is a Carathéodory map and $\mathbb{W} (\cdot, f): T \to E$ is $\Sigma - \mathcal{B}(E)$ measurable, $U(\cdot, \cdot, \mathbb{W} (\cdot, f)): G_S \to E$ is $ \Sigma \otimes \mathcal{B}(E) - \mathcal{B}(E)$ measurable \citep[Lemma 4.52]{RefWorks:RefID:7-aliprantis2006infinite}.

    The rest is almost identical to the proof of Proposition 3  of Hildenbrand (1974, p.60). We first show that for every $c \in \mathbf{R},$ the set 
    $$A = \{t \in T : \mathrm{sup}_{a \in S(t)} U(t, a, \mathbb{W} (t, f)) > c\}$$
    belongs to $\Sigma.$
    This follows from the projection theorem \cite[p.44]{Hildenbrand1977} because $(T, \Sigma, \mu)$ is complete and 
    $$A = \mathrm{proj}_I \{(t,a) \in G_{S}: U(t,a,\mathbb{W} (t, f)) > c \},$$
    where $\mathrm{proj}_I(B)$ means the projection of $B \subset T \times E$ to $T.$
    Since the function $$(t,a) \mapsto U(t,a,\mathbb{W} (t, f)) - \mathrm{sup}_{a \in S(t)}U(t, a, \mathbb{W} (t, f))$$ is $\Sigma \otimes \mathcal{B}(E)-\mathcal{B}(\mathbf{R})$ measurable, The graph
        $$G_{B(\cdot, \mathbb{W} (\cdot, f))}
        = \{(t,a) \in G_{S}: U(t,a,\mathbb{W} (t, f)) = \mathrm{sup}_{a \in S(t)}U(t, a, \mathbb{W} (t, f))\}$$
    belongs to  $\Sigma \times \mathcal{B}(E).$
\end{proof}

Define $\phi: \mathcal{S}_S^1 \twoheadrightarrow \mathcal{S}_S^1$ as $\phi (g) = \{f \in \mathcal{S}_S^1 : f(t) \in B(t, \mathbb{W} (t,g)) \text{ for a.e. } t \in T \}.$ 
Since $S: I \twoheadrightarrow E$ is nonempty, convex, weakly compact valued and integrably bounded, $\mathcal{S}_S^1 \subset L^1(\mu,E)$ is weakly compact by Diestel's theorem \citep[Theorem 3.1]{Yannelis1991}.

\begin{lemma}\label{lem:upper-himicontinuity}
    $\phi$ is nonempty and convex valued and upper hemicontinuous in the weak topology of $\mathcal{S}_S^1.$
\end{lemma}

\begin{proof}
 The nonemptiness follows from Aumann's measurable selection theorem \citep[Theorem 1, p.54]{Hildenbrand1977}, and the convexity follows from the convex valueness of $B(t, \cdot): T \twoheadrightarrow E$.

We prove the upper hemicontinuity of $\phi.$ The idea of the proof is based on the proof of Theorem 5.1 in \cite{RefWorks:RefID:56-khan1991equilibria}.
Since $(T, \Sigma, \mu)$ is essentially countably generated from A.\ref{assum:essentially countably generated}, it follows that $L^1(\mu,E)$ is separable.
As weakly compact sets in separable Banach spaces are metraizable, it follows that $\mathcal{S}_S^1$ is metraizable.
Thus, to prove that $\phi$ is upper hemicontinuous in the weak topology of $E$, it suffices to prove that if $\{f_n\}$ and $\{g_n\}$ are sequences of $\mathcal{S}_S^1$ such that
    \begin{enumerate}
        \item $f_n$ weakly converges to $f \in \mathcal{S}_S^1$,
        \item $g_n$ weakly converges to $g \in \mathcal{S}_S^1$,
        \item $f_n \in \phi(g_n)$ for all $n,$
    \end{enumerate}
then $f \in \phi(g).$

For all $n \in \mathbf{N},$ it holds that $f_n (x) \in B(t, \mathbb{W} (t,g_n)) \text{ for a.e. } t \in T$.  Thus, there exists $T' \in \Sigma$ with $\mu (T \setminus T') = 0$ such that $f_n (t) \in B(t, \mathbb{W} (t,g_n))$ for all $n \in \mathbf{N}$ and $t \in T'.$
Let $A_n = \mathrm{con} ( \bigcup_{k >n }\{f_k\}),$ where $\mathrm{con}(A)$ means the convex hull of a subset $A$ of a vector space. Since $f$ belongs to the weak closure of $A_n,$ it also belongs to the norm closure of $A_n.$
Take a sequence $\{h_\nu^n\}_\nu$ of $A_n$ such that 
\[\lim_{\nu \to \infty} \|h_\nu^n-f\|_1 = 0.\]
Define a sequence $\{\nu_n\}_n$ of natural numbers as follows:
    $$||h_{\nu_1}^1 -f ||_1 < 1,\ ||h_{\nu_2}^2 -f ||_1 < \frac{1}{2},\  \cdots, ||h_{\nu_n}^n -f ||_1 < \frac{1}{n}, \cdots .$$
Define $h_n = h^n_{\nu_n}.$ Then, $h_n \in A_n$ for all $n \in \mathbf{N}$ and $\|h_n -f \|_1 \to 0$.
By taking a subsequence if necessary, we may assume
\(\lim_{n \to \infty} \|h_n(t)-f(t)\| = 0 \text{ for a.e. } t \in T.\)
In particular, we may assume \(\lim_{n \to \infty} \|h_n(t)-f(t)\| = 0 \text{ for all } t \in T'.\)

Select $t \in T'$ arbitrarily, and we prove that $f(t)$ belongs to $B(t, \mathbb{W}(t, g)).$ Suppose not.
Since $B(t, \mathbb{W}(t, g))$ is closed and convex, there exists a bounded linear functional $\Lambda$ on $E$ and $\alpha \in \mathbf{R} $ such that
    $$ \Lambda(b) < \alpha < \Lambda(f(t)) \text{ for all $b \in B(t, \mathbb{W}(t, g)).$}$$
Since $\{g_n\}_n$ weakly converges to $g$, and $\mathbb{W}(t, \cdot): L^1(\mu,E) \to E$ is weakly continuous, $\{\mathbb{W}(t, g_n)\}_n$ weakly converges to $\mathbb{W}(t, g).$
Since $B(t, \cdot)$ is upper hemicontinuous in the weak topology, there exists $ n_0 \in \mathbf{N}$ such that $B(t,\mathbb{W}(t, g_n)) \subset [\Lambda < \alpha]:= \{x \in E: \Lambda(x) < \alpha\}$ for all $n > n_0.$
If $n > n_0,$ we have
$$h_n(t) \in \mathrm{con} (\bigcup_{k > n} \{f_k(t)\}) \subset \mathrm{con} [\bigcup_{k > n} B(t, \mathbb{W}(t, g_n) )] \subset [\Lambda < \alpha].$$
As $n \to \infty,$ we have $\Lambda (h_n(t)) \to \Lambda(f(t)) \leq \alpha.$ A contradiction.
\end{proof}

\begin{proof}[Proof of Theorem \ref{the:main}]\label{sec: proof of example 2}

$\phi: \mathcal{S}_S^1 \twoheadrightarrow \mathcal{S}_S^1$ is nonempty,  convex valued and upper hemicontinuous in the weak topology of $\mathcal{S}_S^1,$ and $\mathcal{S}_S^1$ is nonempty, convex, and weakly compact. From Kakutani-Fan-Gliksberg fixed point theorem, there exists  $f^* \in \mathcal{S}_S^1$ such that $f^* \in \phi (f^*).$ We have
\begin{enumerate}
\item $f^*(t) \in S(t)$ for a.e $t \in T;$
\item
    $U(t, f^*(t), \int_T W(t,s) f^*(s) d \mu (s)) \\
    = \mathrm{max}_{a \in S(t)} U(t, a ,\int_T W(t,s) f^*(s) d \mu (s)) \text{ for a.e. }  t \in T.$
\end{enumerate}
Therefore, $f^*$ is a Nash equilibrium of $\mathcal{G}.$
\end{proof}

\begin{remark}
    The upper hemicontinuity of $\phi$ in Lemma \ref{lem:upper-himicontinuity} can be proved by means of Corollary 4.1 of \cite{Yannelis1991}.
    We denote by $w$-$\mathrm{lim}_{n \to \infty} x_n$ the weak limit of a sequence $\{x_n\}_{n \in \mathbf{N}}$ in $E.$ The \textit{weak upper limit} of a sequence of subsets $\{A_n\}_{n \in N}$ in $E$ is defined by
\begin{equation*}
    w \text{-}\mathrm{Ls}A_n = \{ x \in E: \exists \{x_{n_i}\}_{i \in \mathbf{N}}, x = w\text{-}\mathrm{lim}_{i \to \infty} x_{n_i} 
    \text{ and }
    x_{n_i} \in A_{n_i} \text{ for all } i \in \mathbf{N}\}.
\end{equation*}

Let $\{f_n\}$ and $\{g_n\}$ be as described in Lemma \ref{lem:upper-himicontinuity}. Then, it follows from Corollary 4.1 of \cite{Yannelis1991} that
$f(t) \in \overline{\mathrm{con}}w\text{-}\mathrm{Ls} \{f_n(t)\}$ for a.e. $t \in T,$ where $\overline{\mathrm{con}}(A)$ denotes the closed convex hull of a subset $A$ of a topological vector space.
Since $f_n(x) \in B(t, \mathbb{W}(t,g_n))$ for a.e. $t \in T$ for all $n \in \mathbf{N},$ we have $\overline{\mathrm{con}}w\text{-}\mathrm{Ls} \{f_n(t)\} \subset
\overline{\mathrm{con}}w\text{-}\mathrm{Ls} B(t, \mathbb{W}(t,g_n))$ for a.e. $t \in T$. Finally, since $B(t, \cdot): E \twoheadrightarrow S(t)$ is weakly upper hemicontinuous (and thus has a closed graph), $\mathbb{W}(t, g_n)$ converges to $\mathbb{W}(t, g)$ in the weak topology of $E$, and $B(t, \mathbb{W}(t, g))$ is closed convex, we have
$$f(t) \in \overline{\mathrm{con}}w\text{-}\mathrm{Ls} B(t, \mathbb{W}(t,g_n)) \subset B(t, \mathbb{W}(t, g)) \text{ for a.e. } t\in T.$$
Therefore, it follows that $f \in \phi (g).$
\end{remark}

\subsection{Proof of Proposition \ref{proposition}}
\begin{proof}
Since it is clear in the case of $\lambda = 0,$ we prove the case where $\lambda > 0.$ We may assume without loss of generality that $W(t,s) \leq ||W||_{\infty}$ for all $(t,s) \in T \times T.$

Consider the following integral equation:
        $$\phi (t) = \lambda \int W(t,s) \phi (s) ds + g(t),$$
    where the unknown function \( \phi \) is in \( L^\infty (\mu)\).\footnote{This type of integral equation is called a Fredholm integral equation of the second kind. For more details on this integral equation, see Section 4.11 of \cite{RefWorks:RefID:15-taylor1986introduction}.}
    Define the operator $\mathbb{W}: L^\infty (\mu) \to L^\infty(\mu)$ as $\mathbb{W}f(t) = \int W(t,s) f(s) ds.$ 
    Since $g \in L^\infty (\mu),$ the integral equation can be represented as the following equation in $L^\infty(\mu):$
    $$(I-\lambda \mathbb{W}) \phi = g,$$
    where \(I\) is the identity operator on \(L^\infty (\mu).\)
    
    For all $f \in L^\infty(\mu)$ and $t \in T,$ we have
    $$|\mathbb{W}f(t)| = |\int_T W(t,s) f(s) ds| \leq \int_T |W(t,s) f(s)| ds \leq ||W||_\infty ||f||_\infty.$$
    Thus, $\mathbb{W}$ is a bounded linear operator, and $||\mathbb{W}|| \leq ||W||_\infty,$ where $||\mathbb{W}||$ denotes the operator norm of $\mathbb{W}.$
    Since $||\mathbb{W}|| \leq ||W||_\infty,$ it follows that $\lambda \|\mathbb{W}\| < 1.$ Therefore, according to standard arguments of operator theory, $(I - \lambda \mathbb{W})$ has the inverse operator given by $\sum_{n=0}^\infty \lambda^n \mathbb{W}^n$, namely, the Neumann series of $\mathbb{W}$ \citep[see, for instance,][Theorem 4.1-C, p.164]{RefWorks:RefID:15-taylor1986introduction}.
    Let $\Gamma(\lambda) = \sum_{n =1}^\infty \lambda^{n-1} \mathbb{W}^n$, which yields
        $$(I - \lambda \mathbb{W})^{-1} = \sum_{n=0}^\infty \lambda^n \mathbb{W}^n = I+ \lambda \Gamma(\lambda).$$

\begin{claim}
    $\Gamma(\lambda)$ is represented as:
                    $$(\Gamma(\lambda)f)(t) = \int_T \Gamma(t,s,\lambda) f(s) ds.$$
\end{claim}
\begin{proof}
    First, we prove the convergence of the series $\sum_{n=1}^\infty \lambda^{n-1} W_n(t,s)$. Since $0 \leq W(t, s) \leq ||W||_{\infty},$ it follows by induction that $0 \leq W_n(t,s) \leq ||W||_\infty^n$ for all $(t,s) \in T\times T$ and $n \in \mathbf{N}.$ Hence, we have
    \[
        \sum_{n=1}^\infty \lambda^{n-1} W_n(t,s)
        = \frac{1}{\lambda} \sum_{n=1}^\infty \lambda^{n} W_n(t,s) \\
        \leq \frac{1}{\lambda} \sum_{n=1}^\infty \lambda^{n} ||W||_\infty^n \\
        = \frac{1}{\lambda} \left( \frac{1}{1-\lambda||W||_\infty} -1 \right).
    \]
    
    Therefore, the nonnegative term series $\sum_{n=1}^\infty \lambda^{n-1} W_n(t,s)$ converges to a finite sum. Consequently, for all $(t,s) \in T\times T,$ we have $0 \leq \Gamma (t,s,\lambda) = \sum_{n=1}^\infty \lambda^{n-1} W_n(t,s) \leq \frac{1}{\lambda} (\frac{1}{1-\lambda||W||_\infty} -1).$ 

    Define $\Gamma':L^\infty (\mu) \to L^\infty(\mu)$ as $\Gamma'f(t) = \int \Gamma(t,s,\lambda) f(s) ds.$ We prove that $\Gamma(\lambda) f = \Gamma'f$ for all $f \in L^\infty(\mu).$ Since the set of all simple functions is dense in $L^\infty(\mu),$ and since $\Gamma(\lambda)$ and $\Gamma'$ are continuous and linear, it suffices to consider $f = \chi_A,$ where $A \in \Sigma.$ 

    It is clear from the definition of $W_n$ and Fubini's theorem that $\mathbb{W}^n$ is represented as $\mathbb{W}^n f (t) = \int_T W_n(t,s) f(s) ds.$ Hence, $(\sum_{n=1}^k \lambda^{n-1} \mathbb{W}^n)$ can be written as
    $$ \left( \sum_{n=1}^k \lambda^{n-1} \mathbb{W}^n \right) f (t) = \int_T \left( \sum_{n=1}^k \lambda^{n-1} W_n(t,s) \right) f(s) ds.$$
    Note that
    \begin{equation*}
        \begin{aligned} 
         \left| \Gamma' \chi_A(t) -  \left( \sum_{n=1}^k \lambda^{n-1} \mathbb{W}^n(t,s) \right) \chi_A(t) \right| 
        &= \left| \int_A \left( \Gamma(t,s,\lambda) -  \sum_{n=1}^k \lambda^{n-1} W_n(t,s) \right) ds \right| \\
        &\leq \int_A \left| \Gamma(t,s,\lambda) -  \sum_{n=1}^k \lambda^{n-1} W_n(t,s) \right| ds, 
        \end{aligned}
    \end{equation*}
    and
    \begin{equation*}
        \begin{aligned}
            \left| \Gamma(t,s,\lambda) -  \sum_{n=1}^k \lambda^{n-1} W_n(t,s) \right|
            &= \Gamma(t,s,\lambda) -  \sum_{n=1}^k \lambda^{n-1} W_n(t,s)\\
            &= \sum_{n=k+1}^\infty \lambda^{n-1} W_n(t,s) \\
            &\leq \sum_{n=k+1}^\infty \lambda^{n-1}||W||_\infty^n \\
            &\leq \frac{1}{\lambda}\left( \frac{1}{1-\lambda||W||_\infty} - \frac{1- \lambda^k||W||_\infty^k}{1-\lambda||W||_\infty}\right).    
        \end{aligned} 
    \end{equation*}
Thus, we have
$$\left| \left|\Gamma' \chi_A - \left( \sum_{n=1}^k \lambda^{n-1} \mathbb{W}^n \right) \chi_A \right| \right|_\infty \leq \frac{1}{\lambda}\left( \frac{1}{1-\lambda||W||_\infty} - \frac{1- \lambda^k||W||_\infty^k}{1-\lambda||W||_\infty}\right).$$
As $k$ tends to infinity, the sequence $\{ (\sum_{n=1}^k \lambda^{n-1} \mathbb{W}^n) \chi_A \}_k$ converges to $\Gamma' \chi_A$ in the norm topology of $L^\infty (\mu).$ Since $\{ \sum_{n=1}^k \lambda^{n-1}  \mathbb{W}^n \} _k$ converges to $\Gamma(\lambda)$ in the operator norm, the sequence $\{ (\sum_{n=1}^k \lambda^{n-1} \mathbb{W}^n) \chi_A \}_k$ also converge to $\Gamma(\lambda) \chi_A$ in the norm topology of $L^\infty (\mu).$ Therefore, it follows that $\Gamma(\lambda) \chi_A = \Gamma' \chi_A.$
\end{proof}

Defining $s_g = (I - \lambda \mathbb{W})^{-1} g = (I + \lambda \Gamma(\lambda)) g,$ we obtain $$s_g(t) = g(t) + \lambda \int_T \Gamma (t,s,\lambda) g(s) ds.$$
Since $0 \leq \Gamma(t,s, \lambda) \leq \frac{1}{\lambda} (\frac{1}{1-\lambda||W||_\infty} -1)$ and \(g(t) \in [0,1]\) for a.e. \(t \in T\), we have 
$$0 \leq s_g(t) \leq 1 + \left( \frac{1}{1-\lambda||W||_\infty} -1 \right) =  \frac{1}{1-\lambda||W||_\infty} \text{ for a.e. } t \in T.$$
Hence,
$$\lambda \int W(t,s) s_g(s) + 1 \leq \frac{\lambda}{1-\lambda||W||_\infty} + 1 \ \text{for a.e.} \ t \in T.$$
Since $L$ is assumed to be sufficiently large, we may assume
$$\frac{1}{1-\lambda||W||_\infty}\leq L$$
and
$$\frac{\lambda}{1-\lambda||W||_\infty} + 1 \leq L.$$
In summary, the following three relations hold:
$$s_g(t) \in [0,L] \ \text{for a.e.} \ t \in T.$$
$$0 \leq \lambda \int W(t,s) s_g(s) ds \leq \lambda \int W(t,s) s_g(s) + 1 \leq L \ \text{for a.e.} \ t \in T.$$
$$s_g(t) = \lambda \int W(t,s) s_g(s) + g(t) \in \left[ \int W(t,s) s_g(s) ds, \lambda \int W(t,s) s_g(s) + 1 \right] \text{ for a.e } t \in T.$$
Thus, $s_g$ is a Nash equilibrium of $\mathcal{G}.$

The latter part is clear because $s_g$ is defined as $s_g = (I - \lambda \mathbb{W} )^{-1} g.$
\end{proof}

\subsection{Proof of Lemma \ref{lem:approximation strategy}}\label{sec:Proof of Lemma: approximation strategy}
We begin with the following definition (Definition 7.9 in \cite{RefWorks:RefID:13-rudin1987real}).\footnote{In the definition, \( B(t, r_n) \) represents the open ball centered at \( t \) with radius \( r_n \).}
\begin{definition}
    Suppose $t \in \mathbf{R}$. A sequence $\{E_n\}$ of Borel sets in $\mathbf{R}$ is said to \textit{shrink to} $t$ \textit{nicely} if there is a number $\alpha > 0$ with the following property: There is a sequence $\{B(t,r_n)\}$ of balls with $\lim r_n = 0$ such that $E_n \subset B(t,r_n)$ and
    \begin{equation*}
        \mu (E_n) \geq \alpha \mu (B(t,r_n))
    \end{equation*}
    for all $n \in \mathbf{N}.$
\end{definition}

The following proposition is an extension of a corollary of the Lebesgue differentiation theorem (Theorem 7.10 of \cite{RefWorks:RefID:13-rudin1987real}) to Bochner integrable functions on \((0,1]\).
\begin{proposition}\label{pro:differentiation}
Let $E$ be a Banach space and $f:(0,1] \to E$ be a Bochner integrable function. Associate to each $t \in (0,1]$ a sequence $\{E_n(t)\}$ of subsets of $(0,1]$ that shrinks to $t$ nicely. Then
\[
    f(t) = \lim_{n \to \infty} \frac{1}{\mu (E_n(t))} \int_{E_n(t)} f d \mu
\]
for a.e. $t \in (0,1].$
\end{proposition}

\begin{proof}
    Since 
    \[
        ||f(t) - \frac{1}{\mu (E_n(t))} \int_{E_n(t)} f d \mu|| \leq \frac{1}{\mu (E_n(t))} \int_{E_n(t)} ||f - f(t)|| d \mu
    \]
    for any $t \in (0,1],$
    it suffices to prove that 
    \[
        \lim_{n \to \infty} \frac{1}{\mu (E_n(t))} \int_{E_n(t)} ||f - f(t)|| d \mu = 0
    \]
    for a.e. $t \in (0,1].$

    Since $f$ is Bochner integrable, we may assume $f$ is separably valued. Let $\{a_i\}$ be a dense subset of $f((0,1]).$

    For each $i\in \mathbf{N}$, define $g_i:\mathbf{R} \to \mathbf{R}$ as $g_i(t) = ||f(t) - a_i||$ if $t \in (0,1]$ and $g_i (t)= 0$ otherwise.
    Since $f$ is Bochner integrable, $g_i$ is integrable.
    For each $t \in \mathbf{R}\backslash (0,1],$ define $E_n (t) = B(t,\frac{1}{n}).$ It is clear that $\{E_n(t)\}$ shrinks to $t$ nicely.
    Then, it follows from Theorem 7.10 of \cite{RefWorks:RefID:13-rudin1987real} that
    \begin{equation*}
        \lim_{n \to \infty} \frac{1}{\mu (E_n(t))} \int_{E_n(t)} g_i d \mu = g_i(t)
    \end{equation*}
    for a.e. $t \in \mathbf{R}.$
    Hence we have
    \begin{equation}\label{equ1}
        \lim_{n \to \infty} \frac{1}{\mu (E_n(t))} \int_{E_n(t)} ||f-a_i|| d \mu = ||f(t)-a_i||
    \end{equation}
    for a.e. $t \in (0,1].$
    
    For any $t \in (0,1]$ that satisfies (\ref{equ1}) for all $i \in \mathbf{N},$ it holds that
    \begin{align*}
        \limsup_{n \to \infty} \frac{1}{\mu (E_n(t))} \int_{E_n(t)} ||f - f(t)|| d \mu &\leq 
        \limsup_{n \to \infty} \frac{1}{\mu (E_n(t))} \int_{E_n(t)} ||f - a_i|| d \mu + ||a_i -f(t)||\\
        &= 2||f(t) -a_i||
    \end{align*}
    for all $i \in \mathbf{N}.$
    For any $\epsilon>0,$ take $i \in \mathbf{N}$ such that $||f(t)-a_i|| < \frac{\epsilon}{2}.$ Then we have
    \begin{equation*}
        \limsup_{n \to \infty} \frac{1}{\mu (E_n(t))} \int_{E_n(t)} ||f - f(t)|| d \mu < \epsilon,
    \end{equation*}
    which completes the proof.
\end{proof}

\begin{proof}[Proof of Lemma \ref{lem:approximation strategy}]
    Let $E_{ni} = (\frac{i-1}{k_n}, \frac{i}{k_n}]$ for each $n \in \mathbf{N}$ and each $i =1,2,\cdots, k_n$. Define a $k_n$-step function $f_n:(0,1] \to E$ as
    \begin{equation*}
        f_n (t) = \sum_{i=1}^{k_n} \left(\frac{1}{\mu(E_{ni})} \int_{E_{ni}} f d\mu \right) \chi_{E_{ni}}.
    \end{equation*}
    Since $C$ is closed convex subset of $E$, and $f((0,1]) \subset C,$ it follows from Corollary 8 of [Vector measures] that 
    \[
        \frac{1}{\mu(E_{ni})} \int_{E_{ni}} f d\mu \in C.
    \]
     Hence we have $f_n((0,1]) \subset C$.

    For each $t\in (0,1],$ let $E_n(t)$ be the unique $E_{ni}$ that contains $t.$ 
    Note that $E_n(t) \subset B(t, \frac{1}{k_n}),$ and $\mu(E_n (t)) = \frac{1}{2}\mu (B(t,\frac{1}{k_n}))$ for all $n \in \mathbf{N}$. Hence, since \(k_n \to \infty\), the sequence $\{E_n(t)\}$ shrinks to $t$ nicely.

    Since
    \[
        f_n(t) = \frac{1}{\mu(E_n(t))} \int_{E_n(t)} f d\mu
    \] for all $t \in (0,1]$ and $n \in \mathbf{N}$, it follows from Proposition \ref{pro:differentiation} that $f_n \overset{a.e.}{\longrightarrow} f$.
\end{proof}

\subsection{Proof of Theorem \ref{theorem:limit is Nash}}\label{sec: Proof of limit is Nash}

\begin{proof}
Without loss of generality, we may assume that \( U_n(t) \to U(t) \) and \( f_n(t) \to f(t) \) for all \( t \in (0,1] \). Furthermore, we may assume that \( W_n(t,s) \to W(t,s) \) for all \( (t,s) \in (0,1]^2 \). We define \( E_n \) and \( E \) as follows:
\begin{align*}
E   &= \left\{ t \in (0,1] : U(t)\left(f(t), \int W(t,s)f(s)\, d\mu(s)\right) 
            < \max_{a \in S} U(t)\left(a, \int W(t,s)f(s)\, d\mu(s)\right) \right\}, \\
E_n &= \left\{ t \in (0,1] : U_n(t)\left(f_n(t), \int W_n(t,s)f_n(s)\, d\mu(s)\right) 
            < \max_{a \in S} U_n(t)\left(a, \int W_n(t,s)f_n(s)\, d\mu(s)\right) - \epsilon_n \right\}.
\end{align*}

By assumption, we have \( \mu(E_n) \leq \epsilon_n \). Since \( \epsilon_n \to 0 \), it follows that \( \mu\left(\bigcap_{n>N} E_n\right) = 0 \) for all \( N \in \mathbf{N} \). Therefore, we obtain

$$
\mu\left(\bigcup_{N=1}^\infty \bigcap_{n>N} E_n\right) = 0.
$$

To prove the result, it suffices to show that \( E \subset \bigcup_{N=1}^\infty \bigcap_{n>N} E_n \).
Now, fix \( t \in E \) arbitrarily. To simplify the notation, we introduce the following symbols:
\[
g(s) = W(t,s)f(s), \quad g_n(s) = W_n(t,s)f_n(s);
\]
\[
\alpha = U(t)\left(f(t), \int g(s) \, d\mu(s)\right), \quad \alpha_n = U_n(t)\left(f_n(t), \int g_n(s) \, d\mu(s)\right);
\]
\[
\beta = \max_{a \in S} U(t)\left(a, \int g(s) \, d\mu(s)\right), \quad \beta_n = \max_{a \in S} U_n(t)\left(a, \int g_n(s) \, d\mu(s)\right);
\]
\[
\epsilon = \beta - \alpha > 0.
\]

Since \( S \) is weakly compact, it is norm bounded. Therefore, there exists some \( r > 0 \) such that \( \|f(s)\|, \|f_n(s)\| \leq r \) for all \( s \in (0,1] \) and \( n \in \mathbf{N} \). Moreover, since \( |W_n(t,s)|, |W(t,s)| \leq 1 \), we also have \( \|g(s)\|, \|g_n(s)\| \leq r \) for all \( s \in (0,1] \) and \( n \in \mathbf{N} \).
Since \( \|g_n(s) - g(s)\| \to 0 \) for all \( s \in (0,1] \), it follows from the dominated convergence theorem \citep[][Theorem 3, p.45]{RefWorks:RefID:16-diestel1977vector} that
$$
\int g_n(s)d \mu(s) \to \int g(s)d\mu(s)
$$
in the norm topology, and hence in the weak topology as well.

\begin{claim}\label{claim:continuity of V}
Define \( V: C \times S \times S \to \mathbf{R} \) as \( V(u,a,e) = u(a,e) \). Then, \( V \) is continuous with respect to the weak topology of \( S \). Moreover, \( \max_{a \in S} V(\cdot,a,\cdot) : C \times S \to \mathbf{R} \) is also continuous with respect to the weak topology of \( S \).
\end{claim}

\begin{proof}
Let \( \{(u_\alpha, a_\alpha, e_\alpha)\} \) be a net in \( C \times S \times S \) that converges to \( (u, a, e) \in C \times S \times S \). Note that
\begin{align*}
|V(u,a,e) - V(u_\alpha, a_\alpha, e_\alpha)| 
    &= |u(a,e) - u_\alpha(a_\alpha, e_\alpha)| \\
    &\leq |u(a,e) - u(a_\alpha, e_\alpha)| + |u(a_\alpha, e_\alpha) - u_\alpha(a_\alpha, e_\alpha)| \\
    &\leq |u(a,e) - u(a_\alpha, e_\alpha)| + \|u - u_\alpha\|.
\end{align*}
Since \( u \) is continuous on \( S \times S \), and \( (u_\alpha, a_\alpha, e_\alpha) \to (u, a, e) \), we have \( V(u_\alpha, a_\alpha, e_\alpha) \to V(u,a,e) \).

Since \( S \) is weakly compact, \( \max_{a \in S} V(u,e) \) exists for all \( (u,e) \in C \times S \). Since \( V \) is continuous on \( C \times S \times S \), it follows from Berge’s maximum theorem \citep[][Theorem 17.31]{RefWorks:RefID:7-aliprantis2006infinite} that \( \max_{a \in S} V(\cdot,a,\cdot) : C \times S \to \mathbf{R} \) is continuous.
\end{proof}

Since \( U_n(t) \to U(t) \) and \( f_n(t) \to f(t) \) in the norm topology (and hence in the weak topology), and \( \int g_n(s) d\mu(s) \to \int g(s) d\mu(s) \) in the weak topology, it follows from this claim that \( \alpha_n \to \alpha \) and \( \beta_n \to \beta \).
Since \( \epsilon_n \to 0 \), there exists some \( N \in \mathbf{N} \) such that \( |\alpha - \alpha_n| < \frac{1}{4} \epsilon \), \( |\beta - \beta_n| < \frac{1}{4} \epsilon \), and \( \epsilon_n < \frac{1}{2} \epsilon \) for all \( n > N \). Therefore, we have \( \beta_n - \alpha_n > \epsilon_n \) for all \( n > N \).
This means that
$$
t \in \bigcap_{n > N} E_n \subset \bigcup_{N=1}^\infty \bigcap_{n>N} E_n.
$$
\end{proof}

\subsection{Proof of Theorem \ref{theorem:Nash is limit}}
\label{sec:Proof of Nash is limit}
\begin{proof}
Without loss of generality, we may assume that \( U_n(t) \to U(t) \) and \( f_n(t) \to f(t) \) for all \( t \in (0,1] \). Furthermore, we may assume that \( W_n(t,s) \to W(t,s) \) for all \( (t,s) \in (0,1]^2 \).

Define
\begin{align*}
h_n(t) &= \max_{a \in S} U_n(t)\left(a, \int W_n(t,s) f_n(s) d\mu(s)\right) - U_n(t)\left(f_n(t), \int W_n(t,s) f_n(s) d\mu(s)\right);\\
h(t) &= \max_{a \in S} U(t)\left(a, \int W(t,s) f(s) d\mu(s)\right) - U(t)\left(f(t), \int W(t,s) f(s) d\mu(s)\right);
\end{align*}
\[
\epsilon_n = \inf \left\{ \epsilon > 0 : \mu (\left\{ t \in (0,1] : h_n(t) \leq \epsilon \right\}) \geq 1 - \epsilon \right\}.
\]
It suffices to prove that \( \epsilon_n \to 0 \).

Define \( V: C \times S \times S \to \mathbf{R} \) as \( V(u,a,e) = u(a,e) \). Then, \( h_n \) and \( h \) can be written as:
\begin{align*}
    h_n(t) &= \max_{a \in S} V \left( U_n(t), a, \int W_n(t,s) f_n(s) d\mu(s) \right) - V \left( U_n(t), f_n(t), \int W_n(t,s) f_n(s) d\mu(s) \right),\\
    h(t) &= \max_{a \in S} V \left( U(t), a, \int W(t,s) f(s) d\mu(s) \right) - V\left( U(t), f(t), \int W(t,s) f(s) d\mu(s) \right).
\end{align*}
Note that \( V \) is continuous with respect to the weak topology of \( S \), and \( \max_{a \in S} V(\cdot,a,\cdot) : C \times S \to \mathbf{R} \) is also continuous with respect to the weak topology of \( S \) (see Claim \ref{claim:continuity of V} in the proof of Theorem \ref{theorem:limit is Nash}).

As in the proof of Theorem \ref{theorem:limit is Nash}, we have
\[
\int g_n(s) d\mu(s) \to \int g(s) d\mu(s)
\]
in the weak topology of \(S\). For completeness, we state this again here.
Since \( S \) is weakly compact, it is norm-bounded. Therefore, there exists some \( r > 0 \) such that \( \|f(s)\|, \|f_n(s)\| \leq r \) for all \( s \in (0,1] \) and \( n \in \mathbf{N} \). Moreover, since \( |W_n(t,s)|, |W(t,s)| \leq 1 \), we also have \( \|g(s)\|, \|g_n(s)\| \leq r \) for all \( s \in (0,1] \) and \( n \in \mathbf{N} \).
Since \( \|g_n(s) - g(s)\| \to 0 \) for all \( s \in (0,1] \), it follows from the dominated convergence theorem \citep[][Theorem 3, p.45]{RefWorks:RefID:16-diestel1977vector} that
\[
\int g_n(s) d\mu(s) \to \int g(s) d\mu(s)
\]
in the norm topology, and hence in the weak topology as well.

Since \( U_n(t) \to U(t) \), \( f_n(t) \to f(t) \) in the norm topology (and thus in the weak topology), and \( \int g_n \to \int g \) in the weak topology, it follows from the continuity of \( V \) and \( \max_{a \in S} V(\cdot,a,\cdot) \) that \( h_n(t) \to h(t) \) for all \( t \in (0,1] \). Since \( f \) is a Nash equilibrium of \( \mathcal{G} \), we have \( h(t) = 0 \) for a.e.\ \( t \in (0,1] \). Therefore, \( h_n \overset{a.e.}{\longrightarrow} 0 \)

Since almost everywhere convergence implies convergence in measure in finite measure spaces \citep[][Theorem 13.37]{RefWorks:RefID:7-aliprantis2006infinite}, it follows that \( h_n \to 0 \) in measure. Therefore, for any \( \epsilon > 0 \), there exists some \( N \in \mathbf{N} \) such that for all \( n > N \), it holds that
\[
\mu (\{ t \in (0,1] : h_n(t) > \epsilon \}) \leq \epsilon,
\]
which means
\[
\mu (\{ t \in (0,1] : h_n(t) \leq \epsilon \}) \geq 1 - \epsilon.
\]
Thus, we have \( \epsilon_n < \epsilon \) for all \( n > N \), which implies that
\[
\limsup_{n \to \infty} \epsilon_n \leq \epsilon.
\]
\end{proof}

\bibliographystyle{elsarticle-harv} 

\begin{thebibliography}{37}
\expandafter\ifx\csname natexlab\endcsname\relax\def\natexlab#1{#1}\fi
\providecommand{\url}[1]{\texttt{#1}}
\providecommand{\href}[2]{#2}
\providecommand{\path}[1]{#1}
\providecommand{\DOIprefix}{doi:}
\providecommand{\ArXivprefix}{arXiv:}
\providecommand{\URLprefix}{URL: }
\providecommand{\Pubmedprefix}{pmid:}
\providecommand{\doi}[1]{\href{http://dx.doi.org/#1}{\path{#1}}}
\providecommand{\Pubmed}[1]{\href{pmid:#1}{\path{#1}}}
\providecommand{\bibinfo}[2]{#2}
\ifx\xfnm\relax \def\xfnm[#1]{\unskip,\space#1}\fi
\bibitem[{Aliprantis and Border(2006)}]{RefWorks:RefID:7-aliprantis2006infinite}
\bibinfo{author}{Aliprantis, C.D.}, \bibinfo{author}{Border, K.C.}, \bibinfo{year}{2006}.
\newblock \bibinfo{title}{Infinite Dimensional Analysis: A Hitchhiker's Guide}.
\newblock \bibinfo{publisher}{Springer}.
\bibitem[{Borgs et~al.(2008)Borgs, Chayes, Lovász, Sós and Vesztergombi}]{RefWorks:RefID:410-borgs2008convergent}
\bibinfo{author}{Borgs, C.}, \bibinfo{author}{Chayes, J.T.}, \bibinfo{author}{Lovász, L.}, \bibinfo{author}{Sós, V.T.}, \bibinfo{author}{Vesztergombi, K.}, \bibinfo{year}{2008}.
\newblock Convergent sequences of dense graphs I: Subgraph frequencies, metric properties and testing.
\newblock \bibinfo{journal}{Advances in Mathematics} \bibinfo{volume}{219}, \bibinfo{pages}{1801--1851}.
\bibitem[{Borgs et~al.(2012)Borgs, Chayes, Lovász, Sós and Vesztergombi}]{RefWorks:RefID:411-borgs2012convergent}
\bibinfo{author}{Borgs, C.}, \bibinfo{author}{Chayes, J.T.}, \bibinfo{author}{Lovász, L.}, \bibinfo{author}{Sós, V.T.}, \bibinfo{author}{Vesztergombi, K.}, \bibinfo{year}{2012}.
\newblock Convergent sequences of dense graphs II. Multiway cuts and statistical physics.
\newblock \bibinfo{journal}{Annals of Mathematics} \bibinfo{volume}{176}, \bibinfo{pages}{151--219}.
\bibitem[{Bramoullé and Kranton(2016)}]{10.1093/oxfordhb/9780199948277.013.8}
\bibinfo{author}{Bramoullé, Y.}, \bibinfo{author}{Kranton, R.}, \bibinfo{year}{2016}.
\newblock {Games Played on Networks}, in: \bibinfo{booktitle}{{The Oxford Handbook of the Economics of Networks}}. \bibinfo{publisher}{Oxford University Press}.
\newblock \href{http://arxiv.org/abs/https://academic.oup.com/book/0/chapter/212012097/chapter-ag-pdf/44596635/book\_28058\_section\_212012097.ag.pdf}{{\tt arXiv:https://academic.oup.com/book/0/chapter/212012097/chapter-ag-pdf/44596635/book\_28058\_section\_212012097.ag.pdf}}.
\bibitem[{Carmona and Podczeck(2020)}]{RefWorks:RefID:122-carmona2020pure}
\bibinfo{author}{Carmona, G.}, \bibinfo{author}{Podczeck, K.}, \bibinfo{year}{2020}.
\newblock Pure strategy Nash equilibria of large finite-player games and their relationship to non-atomic games.
\newblock \bibinfo{journal}{Journal of Economic Theory} \bibinfo{volume}{187}, \bibinfo{pages}{105015}.
\bibitem[{Carmona and Podczeck(2021)}]{RefWorks:RefID:394-carmona2021strict}
\bibinfo{author}{Carmona, G.}, \bibinfo{author}{Podczeck, K.}, \bibinfo{year}{2021}.
\newblock Strict pure strategy Nash equilibria in large finite-player games.
\newblock \bibinfo{journal}{Theoretical Economics} \bibinfo{volume}{16}, \bibinfo{pages}{1055--1093}.
\bibitem[{Carmona and Podczeck(2022)}]{RefWorks:RefID:79-carmona2022approximation}
\bibinfo{author}{Carmona, G.}, \bibinfo{author}{Podczeck, K.}, \bibinfo{year}{2022}.
\newblock Approximation and characterization of Nash equilibria of large games.
\newblock \bibinfo{journal}{Economic Theory} \bibinfo{volume}{73}, \bibinfo{pages}{679--694}.
\bibitem[{Diestel and Uhl(1977)}]{RefWorks:RefID:16-diestel1977vector}
\bibinfo{author}{Diestel, J.}, \bibinfo{author}{Uhl, J.J.}, \bibinfo{year}{1977}.
\newblock \bibinfo{title}{Vector measures}.
\newblock \bibinfo{number}{no. 15}, \bibinfo{publisher}{American Mathematical Society}.
\bibitem[{Dunford and Schwartz(1988)}]{dunford1988linear}
\bibinfo{author}{Dunford, N.}, \bibinfo{author}{Schwartz, J.T.}, \bibinfo{year}{1988}.
\newblock \bibinfo{title}{Linear operators, part 1: general theory}. volume~\bibinfo{volume}{10}.
\newblock \bibinfo{publisher}{John Wiley \& Sons}.
\bibitem[{Green(1984)}]{Green1984continuum}
\bibinfo{author}{Green, E.J.}, \bibinfo{year}{1984}.
\newblock Continuum and Finite-Player Noncooperative Models of Competition.
\newblock \bibinfo{journal}{Econometrica} \bibinfo{volume}{52}, \bibinfo{pages}{975--993}.
\bibitem[{He et~al.(2017)He, Sun and Sun}]{RefWorks:RefID:414-he2017modeling}
\bibinfo{author}{He, W.}, \bibinfo{author}{Sun, X.}, \bibinfo{author}{Sun, Y.}, \bibinfo{year}{2017}.
\newblock Modeling infinitely many agents.
\newblock \bibinfo{journal}{Theoretical Economics} \bibinfo{volume}{12}, \bibinfo{pages}{771--815}.
\bibitem[{He and Sun(2022)}]{RefWorks:RefID:70-he2022conditional}
\bibinfo{author}{He, W.}, \bibinfo{author}{Sun, Y.}, \bibinfo{year}{2022}.
\newblock Conditional expectation of Banach valued correspondences and economic applications.
\newblock \bibinfo{journal}{Journal of Mathematical Economics} \bibinfo{volume}{101}, \bibinfo{pages}{102698}.
\bibitem[{Hildenbrand(1974)}]{Hildenbrand1977}
\bibinfo{author}{Hildenbrand, W.}, \bibinfo{year}{1974}.
\newblock \bibinfo{title}{Core and Equilibria of a Large Economy.}
\newblock \bibinfo{publisher}{Princeton University Press}.
\bibitem[{Housman(1988)}]{Hausman1988infinite}
\bibinfo{author}{Housman, D.}, \bibinfo{year}{1988}.
\newblock Infinite Player Noncooperative Games and the Continuity of the Nash Equilibrium Correspondence.
\newblock \bibinfo{journal}{Mathematics of Operations Research} \bibinfo{volume}{13}, \bibinfo{pages}{488--496}.
\bibitem[{Jackson and Zenou(2015)}]{RefWorks:RefID:159-jackson2015chapter}
\bibinfo{author}{Jackson, M.O.}, \bibinfo{author}{Zenou, Y.}, \bibinfo{year}{2015}.
\newblock Chapter 3 - Games on Networks.
\newblock \bibinfo{journal}{Handbook of Game Theory with Economic Applications} \bibinfo{volume}{4}, \bibinfo{pages}{95--163}.
\bibitem[{Janson(2013)}]{RefWorks:RefID:413-janson2013graphons}
\bibinfo{author}{Janson, S.}, \bibinfo{year}{2013}.
\newblock Graphons, cut norm and distance, couplings and rearrangements.
\newblock \bibinfo{journal}{New York journal of mathematics} .
\bibitem[{Kannai(1970)}]{RefWorks:RefID:265-kannai1970continuity}
\bibinfo{author}{Kannai, Y.}, \bibinfo{year}{1970}.
\newblock Continuity Properties of the Core of a Market.
\newblock \bibinfo{journal}{Econometrica} \bibinfo{volume}{38}, \bibinfo{pages}{791--815}.
\bibitem[{Keisler and Sun(2009)}]{RefWorks:RefID:28-keisler2009saturated}
\bibinfo{author}{Keisler, H.J.}, \bibinfo{author}{Sun, Y.}, \bibinfo{year}{2009}.
\newblock Why saturated probability spaces are necessary.
\newblock \bibinfo{journal}{Advances in Mathematics} \bibinfo{volume}{221}, \bibinfo{pages}{1584--1607}.
\bibitem[{Khan et~al.(2013a)Khan, Rath, Sun and Yu}]{RefWorks:RefID:273-khan2013large}
\bibinfo{author}{Khan, M.A.}, \bibinfo{author}{Rath, K.P.}, \bibinfo{author}{Sun, Y.}, \bibinfo{author}{Yu, H.}, \bibinfo{year}{2013}a.
\newblock Large games with a bio-social typology.
\newblock \bibinfo{journal}{Journal of Economic Theory} \bibinfo{volume}{148}, \bibinfo{pages}{1122--1149}.
\bibitem[{Khan et~al.(2013b)Khan, Rath, Yu and Zhang}]{RefWorks:RefID:424-khan2013large}
\bibinfo{author}{Khan, M.A.}, \bibinfo{author}{Rath, K.P.}, \bibinfo{author}{Yu, H.}, \bibinfo{author}{Zhang, Y.}, \bibinfo{year}{2013}b.
\newblock Large distributional games with traits.
\newblock \bibinfo{journal}{Economics Letters} \bibinfo{volume}{118}, \bibinfo{pages}{502--505}.
\bibitem[{Khan and Sun(1999)}]{RefWorks:RefID:423-khan1999non-cooperative}
\bibinfo{author}{Khan, M.A.}, \bibinfo{author}{Sun, Y.}, \bibinfo{year}{1999}.
\newblock Non-cooperative games on hyperfinite Loeb spaces.
\newblock \bibinfo{journal}{Journal of Mathematical Economics} \bibinfo{volume}{31}, \bibinfo{pages}{455--492}.
\bibitem[{Khan and Yannelis(1991)}]{RefWorks:RefID:56-khan1991equilibria}
\bibinfo{author}{Khan, M.A.}, \bibinfo{author}{Yannelis, N.C.}, \bibinfo{year}{1991}.
\newblock Equilibria in Markets with a Continuum of Agents and Commodities, in: \bibinfo{editor}{Khan, M.A.}, \bibinfo{editor}{Yannelis, N.C.} (Eds.), \bibinfo{booktitle}{Equilibrium Theory in Infinite Dimensional Spaces}. \bibinfo{publisher}{Springer Berlin Heidelberg}, \bibinfo{address}{Berlin, Heidelberg}, pp. \bibinfo{pages}{233--248}.
\bibitem[{Kunszenti-Kovács(2019)}]{RefWorks:RefID:416-kunszenti-kovács2019uniqueness}
\bibinfo{author}{Kunszenti-Kovács, D.}, \bibinfo{year}{2019}.
\newblock Uniqueness of Banach space valued graphons.
\newblock \bibinfo{journal}{Journal of Mathematical Analysis and Applications} \bibinfo{volume}{474}, \bibinfo{pages}{413--440}.
\bibitem[{Kunszenti-Kovács et~al.(2022)Kunszenti-Kovács, Lovász and Szegedy}]{RefWorks:RefID:415-kunszenti-kovács2022multigraph}
\bibinfo{author}{Kunszenti-Kovács, D.}, \bibinfo{author}{Lovász, L.}, \bibinfo{author}{Szegedy, B.}, \bibinfo{year}{2022}.
\newblock Multigraph limits, unbounded kernels, and Banach space decorated graphs.
\newblock \bibinfo{journal}{Journal of Functional Analysis} \bibinfo{volume}{282}, \bibinfo{pages}{109284}.
\bibitem[{Lovász(2012)}]{RefWorks:RefID:412-lovász2012large}
\bibinfo{author}{Lovász, L.}, \bibinfo{year}{2012}.
\newblock \bibinfo{title}{Large networks and graph limits}. volume~\bibinfo{volume}{60}.
\newblock \bibinfo{publisher}{American Mathematical Soc.}
\bibitem[{Lovász and Szegedy(2006)}]{RefWorks:RefID:409-lovász2006limits}
\bibinfo{author}{Lovász, L.}, \bibinfo{author}{Szegedy, B.}, \bibinfo{year}{2006}.
\newblock Limits of dense graph sequences.
\newblock \bibinfo{journal}{Journal of Combinatorial Theory, Series B} \bibinfo{volume}{96}, \bibinfo{pages}{933--957}.
\bibitem[{Mas-Colell(1984)}]{RefWorks:RefID:402-mas-colell1984theorem}
\bibinfo{author}{Mas-Colell, A.}, \bibinfo{year}{1984}.
\newblock On a theorem of Schmeidler.
\newblock \bibinfo{journal}{Journal of Mathematical Economics} \bibinfo{volume}{13}, \bibinfo{pages}{201--206}.
\bibitem[{Parise and Ozdaglar(2023)}]{parise2023}
\bibinfo{author}{Parise, F.}, \bibinfo{author}{Ozdaglar, A.}, \bibinfo{year}{2023}.
\newblock Graphon Games: A Statistical Framework for Network Games and Interventions.
\newblock \bibinfo{journal}{Econometrica} \bibinfo{volume}{91}, \bibinfo{pages}{191--225}.
\bibitem[{Qiao and Yu(2014)}]{RefWorks:RefID:281-qiao2014space}
\bibinfo{author}{Qiao, L.}, \bibinfo{author}{Yu, H.}, \bibinfo{year}{2014}.
\newblock On the space of players in idealized limit games.
\newblock \bibinfo{journal}{Journal of Economic Theory} \bibinfo{volume}{153}, \bibinfo{pages}{177--190}.
\bibitem[{Qiao et~al.(2016)Qiao, Yu and Zhang}]{RefWorks:RefID:279-qiao2016closed-graph}
\bibinfo{author}{Qiao, L.}, \bibinfo{author}{Yu, H.}, \bibinfo{author}{Zhang, Z.}, \bibinfo{year}{2016}.
\newblock On the closed-graph property of the Nash equilibrium correspondence in a large game: A complete characterization.
\newblock \bibinfo{journal}{Games and Economic Behavior} \bibinfo{volume}{99}, \bibinfo{pages}{89--98}.
\bibitem[{Rokade and Parise(2023)}]{Rokade2023}
\bibinfo{author}{Rokade, K.}, \bibinfo{author}{Parise, F.}, \bibinfo{year}{2023}.
\newblock Graphon Games with Multiple Nash Equilibria: Analysis and Computation. Available at SSRN.
\newblock \URLprefix \url{https://ssrn.com/abstract=4354931}, \DOIprefix\doi{http://dx.doi.org/10.2139/ssrn.4354931}.
\bibitem[{Rudin(1987)}]{RefWorks:RefID:13-rudin1987real}
\bibinfo{author}{Rudin, W.}, \bibinfo{year}{1987}.
\newblock \bibinfo{title}{Real and complex analysis}.
\newblock \bibinfo{number}{. Mathematics series}, \bibinfo{publisher}{McGraw-Hill}.
\bibitem[{Schmeidler(1973)}]{RefWorks:RefID:419-schmeidler1973equilibrium}
\bibinfo{author}{Schmeidler, D.}, \bibinfo{year}{1973}.
\newblock Equilibrium points of nonatomic games.
\newblock \bibinfo{journal}{Journal of Statistical Physics} \bibinfo{volume}{7}, \bibinfo{pages}{295--300}.
\bibitem[{Sun and Zhang(2015)}]{RefWorks:RefID:274-sun2015pure-strategy}
\bibinfo{author}{Sun, X.}, \bibinfo{author}{Zhang, Y.}, \bibinfo{year}{2015}.
\newblock Pure-strategy Nash equilibria in nonatomic games with infinite-dimensional action spaces.
\newblock \bibinfo{journal}{Economic Theory} \bibinfo{volume}{58}, \bibinfo{pages}{161--182}.
\bibitem[{Taylor and Lay(1986)}]{RefWorks:RefID:15-taylor1986introduction}
\bibinfo{author}{Taylor, A.E.}, \bibinfo{author}{Lay, D.C.}, \bibinfo{year}{1986}.
\newblock \bibinfo{title}{Introduction to functional analysis}.
\newblock \bibinfo{edition}{2nd ed., reprint ed} ed., \bibinfo{publisher}{R.E. Krieger Pub. Co.}
\bibitem[{Wu(2022)}]{RefWorks:RefID:393-wu2022pure-strategy}
\bibinfo{author}{Wu, B.}, \bibinfo{year}{2022}.
\newblock On pure-strategy Nash equilibria in large games.
\newblock \bibinfo{journal}{Games and Economic Behavior} \bibinfo{volume}{132}, \bibinfo{pages}{305--315}.
\bibitem[{Yannelis(1991)}]{Yannelis1991}
\bibinfo{author}{Yannelis, N.C.}, \bibinfo{year}{1991}.
\newblock Integration of Banach-Valued Correspondence, in: \bibinfo{editor}{Khan, M.A.}, \bibinfo{editor}{Yannelis, N.C.} (Eds.), \bibinfo{booktitle}{Equilibrium Theory in Infinite Dimensional Spaces}. \bibinfo{publisher}{Springer Berlin Heidelberg}, \bibinfo{address}{Berlin, Heidelberg}, pp. \bibinfo{pages}{2--35}.

\end{thebibliography}





\end{document}